\documentclass[authoryear, preprint,12pt]{elsarticle}

\usepackage{amsmath}

\usepackage{amssymb}
\usepackage[colorlinks=true,citecolor=blue, linkcolor=blue]{hyperref}
\usepackage[normalem]{ulem}
\usepackage{subcaption}
\usepackage{amsthm}

\usepackage{xcolor}
\usepackage{float}
 \usepackage{lineno}
\journal{arXiv}

\newtheorem{lemma}{\bf Lemma}[]

\numberwithin{equation}{section}

\begin{document}
\begin{frontmatter}

%% Title, authors and addresses

%% use the tnoteref command within \title for footnotes;
%% use the tnotetext command for theassociated footnote;
%% use the fnref command within \author or \address for footnotes;
%% use the fntext command for theassociated footnote;
%% use the corref command within \author for corresponding author footnotes;
%% use the cortext command for theassociated footnote;
%% use the ead command for the email address,
%% and the form \ead[url] for the home page:
%% \title{Title\tnoteref{label1}}
%% \tnotetext[label1]{}
%% \author{Name\corref{cor1}\fnref{label2}}
%% \ead{email address}
%% \ead[url]{home page}
%% \fntext[label2]{}
%% \cortext[cor1]{}
%% \affiliation{organization={},
%%             addressline={},
%%             city={},
%%             postcode={},
%%             state={},
%%             country={}}
%% \fntext[label3]{}

\title{Coupling plankton and cholera dynamics: insights into outbreak prediction and practical disease management}
\author[inst1]{Biplab Maity}
\author[inst2,insta3,inst1]{Swarnendu Banerjee\corref{corauthor}}
\author[inst4,inst1]{Abhishek Senapati}
\author[inst5,inst6]{Jon Pitchford}
\author[inst1]{Joydev Chattopadhyay}
\cortext[corauthor]{Corresponding author\\ Email: swarnendubanerjee92@gmail.com}

 \address[inst1]{Agricultural and Ecological Research Unit, Indian Statistical Institute, 203, B. T. Road, Kolkata 700108, India}
 \address[inst2]{Dutch Institute for Emergent Phenomenon, Institute for Biodiversity and Ecosystem Dynamics, University of Amsterdam, 1090 GE Amsterdam, The Netherlands}
 \address[insta3]{Copernicus Institute of Sustainable Development, Utrecht University, 3508 TC, Utrecht, the Netherlands}
\address[inst4]{Saw Swee Hock School of Public Health, National University of Singapore, Singapore}
\address[inst5]{Department of Biology, University of York, Wentworth Way, York, YO10 5DD, UK}
\address[inst6]{Department of Mathematics, University of York, Heslington, York, YO10 5DD, UK}

\begin{abstract}
Despite extensive control efforts over the centuries, cholera remains a globally significant health issue. Seasonal emergence of cholera cases has been reported, particularly in the Bengal delta region, which is often synchronized with plankton blooms. This phenomenon has been widely attributed to the commensal interaction between \textit{Vibrio cholerae} and zooplankton in aquatic environments. Understanding the role of plankton dynamics in cholera epidemiology is therefore crucial for effective policy-making. To this end, we propose and analyze a novel compartment-based transmission model that integrates phytoplankton-zooplankton interactions into a human-bacteria cholera model. We show that zooplankton-mediated transmission can lead to counterintuitive outcomes, such as an outbreak with a delayed and lower peak still resulting in a larger overall outbreak size. 
% We show that zooplankton-mediated transmission can lead to counterintuitive outcomes such as an outbreak with a delayed and lower peak, can still result in a larger outbreak size. 
Such outbreaks are prolonged by maintaining infections at lower levels during the post-peak phase, thereby intensifying epidemic overshoot and promoting the inter-epidemic persistence of pathogens. Furthermore, our analysis reveals that the timing of filtration-like interventions can be strategically guided by ecological indicators, such as phytoplankton blooms. Our study underscores the importance of incorporating ecological aspects in epidemiological research to gain a deeper understanding of disease dynamics.

\end{abstract}

% To this end, we propose and analyze a novel compartment-based transmission model that integrates phytoplankton-zooplankton interactions into a human-bacteria cholera model. We show that zooplankton-mediated transmission can lead to counterintuitive outcomes, where R0 may not always be a reliable indicator of disease severity. For the same R0, an outbreak, despite a delayed and lower peak, can still result in a larger outbreak size. Such outbreaks are prolonged by maintaining infections at lower levels during the post-peak phase, thereby intensifying epidemic overshoot and promoting the inter-epidemic persistence of pathogens. \sout{The outbreak, in such a case, is prolonged by maintaining infections at a lower level during the post-peak phase, thereby intensifying the epidemic overshoot and promoting inter-epidemic persistence of pathogens.}


\end{frontmatter}

\section{Introduction}
Cholera, while treatable, remains a major global health emergency, with an estimated 2.9 million annual cases including 95,000-1,07,000 deaths globally \citep{troeger2018estimates, cholera_global}. Despite uncertainties in case reporting, recent methodologies indicate between 470,000-790,00 cholera cases, and up to 5,000 deaths, annually in the early 2020s \citep{venkatesan2024new, cholera12}. The majority of these cases are in Africa, but the presence of cholera in Afghanistan, Yemen, Pakistan, Haiti, Bangladesh, and India underlines its importance in emerging and developing countries \citep{ilic2023global,xu2024enhanced, amisu2024cholera, cholera_africa, cholera_ECDC}. Country-specific factors and a critical shortage of oral vaccines in 2023 have led the WHO to categorize the resurgence of cholera as a grade 3 emergency \citep{cholera12}.

The consistent seasonal re-emergence of cholera in the Bengal Delta region over the past decades has been attributed to seasonal plankton blooms \citep{islam2015role, de2011role, huq1983ecological, huq2005critical, jutla2012satellite}. It has been widely established that \textit{Vibrio cholerae} (\textit{V. cholerae}), the causative bacteria, is associated with plankton \citep{vezzulli2010environmental, colwell1996global, almagro2013cholera, lutz2013environmental}. The detection of \textit{V. cholerae} associated with zooplankton cells along the coasts of Brazil and Mexico suggests that the ecological relationship between bacteria and plankton is widespread \citep{martinelli2010vibrio, lizarraga2009association}. Copepods, a crustacean zooplankton, are the largest known natural reservoir of the \textit{V. cholerae} and have been implicated as a potential vector \citep{vezzulli2010environmental, colwell1994environmental}. Experimental studies by \cite{turner2009plankton} and \cite{rawlings2007association} in natural estuarine and coastal ecosystems of Georgia, USA, as well as in the Bengal delta region found a significant correlation between \textit{V. cholerae} density and copepod abundance. Studies indicate that a single copepod can carry $10^4$ to $10^6$ bacterial cells \citep{de2008environmental, heidelberg2002bacteria, colwell1996viable, huq1983ecological}. 

\textit{Vibrios} colonize the gut of the zooplankton and the surface biofilms, where they can multiply rapidly under favorable nutrient conditions \citep{lutz2013environmental,sochard1979bacteria, huq1983ecological, huq1984influence, kirn2005colonization}. Further, the zooplankton protects the \textit{Vibrios} from grazing as well as from chemical disinfectants, significantly prolonging bacterial survival compared to free-living cells \citep{conner2016staying, huq1983ecological, chowdhury1997effect, tang2010linkage, perera2022zooplankton}. This relationship between the zooplankton and the \textit{Vibrios}, where only the former benefits from the association, is known as commensalism \citep{lutz2013environmental}. The empirical study by \cite{lipp2002effects} asserted that, at high concentrations, \textit{Vibrios} are more likely to attach to zooplankton cells rather than remain free in the surrounding water. As a result, exposure to commensal zooplankton could potentially lead to the consumption of a large bacterial inoculum. This underscores the fundamental role of zooplankton in cholera dynamics and emphasizes the importance of studying plankton ecology, particularly in regions with evidence of \textit{V. cholerae} reservoirs \citep{constantin2014community}.

Mathematical models have been proven to be an important and reliable tool for understanding cholera transmission mechanisms and guiding policymaking during several outbreaks \citep{king2008inapparent, fung2014cholera, mukandavire2011estimating, miller2010modeling, andrews2011transmission, maity2023model}. In spite of numerous empirical studies over the past decades linking plankton abundance to increased cholera infections \citep{islam2015role, de2011role, lipp2002effects, huq1983ecological, huq2005critical, jutla2012satellite, islam1994probable, kirn2005colonization}, a mechanistic model that explicitly connects the ecology of plankton to \textit{V. cholerae} transmission is still lacking. A recent study by \cite{kolaye2019mathematical} considered commensalism between bacteria and phytoplankton, focusing on bacterial metabolism changes. However, this study neither included any explicit compartment for commensal plankton nor studied transmission via plankton. To address this gap, we have developed a novel compartmental transmission model that integrates the phytoplankton-zooplankton interactions with the classical susceptible-infected-recovered-bacteria (SIRB) cholera model \citep{Codeco01}. Our model includes a separate compartment for \textit{Vibrio}-associated zooplankton cells and accounts for two transmission routes: one for free-living bacterial cells and another for bacteria-associated zooplankton. The role of zooplankton-mediated transmission in short-term epidemic outbreaks has been investigated. We assess the relative importance of transmission routes in shaping epidemic progression. Additionally, we have explored the epidemic overshoot phenomenon, which is linked to the post-peak severity of outbreaks. Further, we evaluate the effect of filtration as a practical disease management measure.

\section{Model formulation} \label{sec:model_formulation}
Our approach combines established SIR models for disease transmission with well-developed models for plankton dynamics, using a minimal set of biologically justifiable and quantifiable assumptions. Fig.~\ref{fig1-schematic} and Table~\ref{table-1} summaries this approach. This framework facilitates a range of analyses, including the characterization of steady states and transient dynamics, computational sensitivity analysis, and the evaluation of control scenarios, across a range of ecologically relevant time scales.

The total human population at time $t$, $N(t)$, is divided into three compartments: susceptible $(S(t))$, infected $(I(t))$ and recovered $(R(t))$. Earlier models have typically included a single compartment for the bacterial population in the water column \citep{Codeco01, righetto2012role}. However, the commensal relationship between \textit{V. cholerae} and zooplankton results in the bacteria being associated with zooplankton cells, which can facilitate human infection upon exposure \citep{vezzulli2010environmental, perera2022zooplankton}.
% Human exposure to these zooplankton cells can potentially cause infection \citep{vezzulli2010environmental, perera2022zooplankton}. 
This phenomenon motivates the inclusion of an explicit compartment for bacteria-associated zooplankton ($Z_B(t)$), which is formed when free-living bacteria ($B(t)$) attaches to uncolonized zooplankton ($Z_F(t)$). Hence, $Z(t)=Z_F(t)+Z_B(t)$ is the total zooplankton density at time $t$. The force of \textit{Vibrio}-zooplankton commensalism is represented by the term $\displaystyle c \sigma\frac{B}{h_m + B}$. Here, $\sigma$ denotes the rate of bacteria-zooplankton association and $h_m$ is the half-saturation constant of the association. The parameter $c$ represents the average number of \textit{Vibrio} cells per zooplankton, which we term the `colonization coefficient'. Since phytoplankton ($P(t)$) governs zooplankton ($Z(t)$) abundance, we consider the dynamics of phytoplankton-zooplankton using a well-studied Rosenzweig–MacArthur type model \citep{rosenzweig1971paradox}.

Susceptible humans become infected through exposure to free-living bacteria (via water contamination), $B(t)$, with a force of infection $\lambda_B=\displaystyle \frac{\beta B}{h_b+B}$ and \textit{Vibrio}-associated zooplankton, $Z_B(t)$, with a force of infection $\lambda_Z=\displaystyle \frac{\beta_z Z_B}{h_z+Z_B}$. Here, $h_b$ and $h_z$ represent the half-saturation constant for transmission via bacterial and zooplankton routes, respectively. Note that $h_z$ depends on the pathogen load of commensal zooplankton and is inversely proportional to the bacterial colonization coefficient, $c$. The use of Holling type-II transmission terms follows existing literature \citep{Codeco01, hartley2005hyperinfectivity, righetto2012role}. Our model, which includes transmission via both free-living bacteria and bacteria-associated zooplankton, is given as follows (see Fig.~\ref{fig1-schematic} for schematic).

\begin{eqnarray}
\begin{array}{llll}

\displaystyle \frac{d S}{dt} &=& \displaystyle \Lambda -  \frac{\beta SB}{h_b + B} - \frac{\beta_z SZ_B}{h_z + Z_B}-\mu S + \omega R,\\[3ex]

\displaystyle \frac{d I}{dt} &=& \displaystyle \frac{\beta SB}{h_b + B} + \frac{\beta_z SZ_B}{h_z + Z_B}  - (\gamma + \mu + \delta)I,\\[3ex] 

\displaystyle \frac{d R}{dt} &=& \displaystyle \gamma I -(\mu + \omega)R,\\[3ex]

\displaystyle \frac{d B}{dt} &=& \displaystyle  \xi I -d_b B-c \sigma \frac{BZ_F}{h_m+B},\\[3ex]

\displaystyle \frac{d Z_B}{dt} &=& \displaystyle  \sigma \frac{BZ_F}{h_m+B} -d_zZ_B,\\[3ex]

\displaystyle \frac{d Z_F}{dt} &=& \displaystyle \eta \frac{\alpha P(Z_F+Z_B)}{h_p +P} -d_zZ_F-\sigma \frac{BZ_F}{h_m+B},\\[3ex]

\displaystyle \frac{d P}{dt} &=& \displaystyle r_pP(1-\frac{P}{K})-\frac{\alpha P(Z_F+Z_B)}{h_p+P}.

\end{array}
\label{final_model}
\end{eqnarray}

Infected individuals either recover at a rate $\gamma$ or die due to infection at a rate $\delta$ and contribute bacterial cells into the environment through excretion at a rate $\xi$ over their infectious period. As cholera does not confer life-long immunity \citep{sanches2011role, righetto2012role}, recovered individuals, although initially immune to the pathogen, eventually lose immunity at a rate $\omega$ and replenish into the susceptible compartment. $\Lambda$ denotes the recruitment rate of susceptible individuals through birth and immigration, while each class experiences a natural death rate $\mu$. Here, $d_b$ represents the rate of loss of free-living bacteria, accounting for both mortality and the decline in vitality. Phytoplankton grows at a rate $r_p$ up to its maximum achievable density, known as the carrying capacity, $K$. The maximal grazing rate of the zooplankton on phytoplankton is $\alpha$ with a conversion coefficient $\eta$. Here, $d_z$ and $h_p$ refer to the zooplankton death rate and the half-saturation concentration of zooplankton grazing, respectively.

Note that the plankton dynamics, represented by the last three equations in model~\eqref{final_model}, is independent of the human-bacteria (SIRB) dynamics (see Eq.~\eqref{phyto_zoo_model} in \ref{appendix-phyto-zoo model}).
% i.e., the bacteria-zooplankton association does not affect the plankton density which is representative of the commensal interaction. 
This independence implies that the bacteria-zooplankton association does not influence overall plankton density, which is representative of the commensal interaction. However, this association leads to the formation of $Z_B$ and reduces free-living bacterial density, thereby affecting human-bacteria dynamics. Note that we do not consider direct human-to-human transmission here, as it has been shown that the aquatic environment plays a decisive role in the survival and transmission of pathogens during cholera outbreaks in several regions \citep{vezzulli2010environmental}. If the bacteria-zooplankton association is excluded (i.e. $\sigma=0$), our model~\eqref{final_model} resembles previous SIRB cholera models, such as those used in \citep{Codeco01,sanches2011role}.

%%%%%%%%%%%%%%%%%%%%%%%%%%%    Figure 1: schematic
%%%%%%%%%%%%%%%%%%%%%%%%%%%
\begin{figure}[H]
 \begin{center}
{\includegraphics[width=0.8\textwidth]{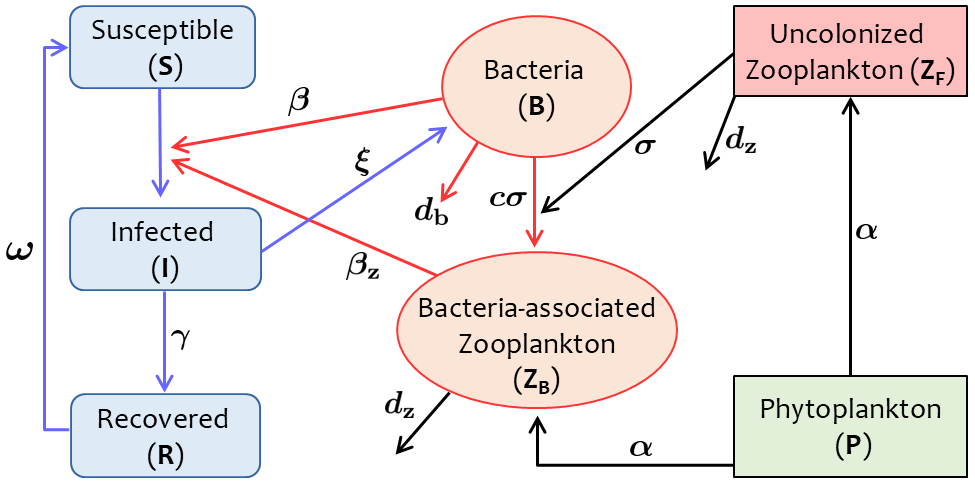}}
 \end{center}
 \caption{Model schematic illustrating the linkage between phytoplankton-zooplankton interactions and the classical susceptible-infected-recovered-bacteria (SIRB) cholera model through the commensal association between \textit{Vibrio cholerae} and zooplankton. For clarity, human birth and death terms are not shown. Model parameters and their corresponding values are given in Table~\ref{table-1}.}
 \label{fig1-schematic}
 \end{figure}

%%%%%%%%%%%%%%%%%%%%%    Table 1: parameters_table
%%%%%%%%%%%%%%%%%%%%%
\begin{table}[h!]
\tabcolsep 14 pt
\centering
\begin{tabular}{|p{1 cm}| p{4cm}|p{2.6cm}| p{2.5cm}| p{3 cm}|}
\hline
\tiny{\textbf{Parameters}} & \tiny{\textbf{Description}} &\tiny{\textbf{Value}} & \tiny{\textbf{Reference}}  \\
[0.5ex]
\hline
\tiny{$\Lambda$} & \tiny{Constant recruitment rate of human population} & \tiny{$6.85\times 10^{-5}\times N_0$ person $\rm day^{-1}$} & \tiny{\citep{dimitrov2014comparative}}\\

\tiny{$\beta$} & \tiny{Transmission rate via bacteria} & \tiny{0.214 $\rm day^{-1}$} &\tiny{\citep{hartley2005hyperinfectivity, neilan2010modeling, mukandavire2011estimating}} \\

\tiny{$\beta_z$} & \tiny{Transmission rate via zooplankton} & \tiny{variable but less than $\beta$} &\tiny{Investigated (\ref{param})} \\

\tiny{$h_b$} & \tiny{Half-saturation constant of bacterial transmission} & \tiny{$10^9$ cells $\rm L^{-1}$}&\tiny{\citep{hartley2005hyperinfectivity}} \\

\tiny{$h_z$} & \tiny{Half-saturation constant of transmission via zooplankton} & \tiny{20 $\rm (mg~dw)$ $\rm L^{-1}$} &\tiny{Investigated (\ref{param})} \\

\tiny{$\mu$} & \tiny{Natural death rate of human} & \tiny{$3.8 \times 10^{-5}$ $\rm day^{-1}$} &\tiny{\citep{keeling2008modeling,martcheva2015introduction}}\\

\tiny{$\omega$} & \tiny{Rate of immunity loss of recovered individuals} & \tiny{$0.00092$ $\rm day^{-1}$}& \tiny{\citep{koelle2005refractory, dimitrov2014comparative}}\\

\tiny{$\gamma$} & \tiny{Recovery rate of infected human} & \tiny{$0.2~\rm day^{-1}$} & \tiny{\citep{dimitrov2014comparative, mukandavire2011estimating}}\\

\tiny{$\delta$} & \tiny{Disease induced mortality rate of human} &\tiny{0.013 $\rm day^{-1}$}& \tiny{\citep{kolaye2019mathematical, king2008inapparent, neilan2010modeling}}\\

\tiny{$\xi$} & \tiny{Bacteria shedding rate of infected humans} &\tiny{10-10$^4$~cells~$\rm L^{-1}$ $\rm day^{-1}$ per person}& \tiny{\citep{neilan2010modeling, Codeco01, jensen2006modeling}}\\

\tiny{$d_b$} & \tiny{Removal rate of the bacteria (birth - death)} & \tiny{0.33 $\rm day^{-1}$} & \tiny{\citep{Codeco01, kolaye2019mathematical}}  \\

\tiny{$\sigma$} & \tiny{Rate of bacteria-zooplankton association} & \tiny{0.005-0.1 $\rm day^{-1}$}& \tiny{Investigated (Sec~\ref{sec:result_outbreak})}\\

\tiny{$c$} & \tiny{Colonization coefficient of bacteria} & \tiny{$5\times 10^7$ $\rm cells~(mg~dw)^{-1}$}& \tiny{Investigated (\ref{param})}\\

\tiny{$h_m$} & \tiny{Half saturation constant for bacteria-zooplankton association} & \tiny{$10^7$ cells $\rm L^{-1}$}& \tiny{\ref{param}}\\

\tiny{$\eta$} & \tiny{Conversion coefficient} & \tiny{0.6} & \tiny{\citep{scheffer1991fish}}  \\

\tiny{$\alpha$} & \tiny{Predation rate} & \tiny{$0.4$ $\rm day{-1}$} & \tiny{\citep{scheffer1991fish}}  \\

\tiny{$h_p$} & \tiny{Half saturation constant of phytoplankton} & \tiny{0.6 (mg~dw) $\rm L^{-1}$} & \tiny{\citep{scheffer1991fish}}\\

\tiny{$d_z$} & \tiny{Death rate of zooplankton} & \tiny{0.06 $\rm day^{-1}$} &  \tiny{\citep{da2020non,hirst2002mortality,di2017non}}\\

\tiny{$r_p$} & \tiny{Growth rate of phytoplankton} & \tiny{0.5 $\rm day^{-1}$} & \tiny{\citep{scheffer1991fish}}  \\

\tiny{$K$} & \tiny{Carrying capacity of phytoplankton} & \tiny{1 (mg~dw) $\rm L^{-1}$} & \tiny{\citep{freund2006bloom}}\\[1.5ex]

\hline
\end{tabular}
\caption{\small{Description of parameters for the model~\eqref{final_model}}}
\label{table-1}
\end{table}

\newpage
\section{Results}
\label{section_results}
Our results are presented according to the ecologically and management-motivated time scales relevant to this study. The main body of work (Sec.~\ref{sec:result_outbreak}) analyses cholera outbreaks in human populations over a short time scale, where human population and immunity dynamics can be neglected (i.e. $\mu, \Lambda = 0$ and $\omega = 0$ in model~\eqref{final_model}). We then contextualize these findings in the context of practical interventions connected to water filtration (Sec.~\ref{sec:result_filtration}).

\subsection{Outbreak dynamics}\label{sec:result_outbreak}
To remove the effect of transient plankton dynamics on disease outbreak, we consider the coexistent steady-state plankton densities as the initial condition for phytoplankton ($P_0=P^*$) and zooplankton ($Z_{F_0}=Z^*, Z_B=0$) in model~\eqref{final_model} (see Eq.~\eqref{P_and_Z}  for $P^*, Z^*$ in \ref{appendix-phyto-zoo model}). This approach is reasonable in light of the following: (i) the transient plankton dynamics has no biological significance in the context of an outbreak, and (ii) the outbreak does not alter the plankton dynamics. Note that the zooplankton-free equilibrium ($Z_F=Z_B=0$) of the phytoplankton-zooplankton system is not relevant in the context of our study.

% \subsubsection{Condition for initial outbreak growth}
The necessary condition for the initial growth of an outbreak in the presence of both transmission routes is given by (see \ref{out_cond_cal})
\begin{eqnarray}
    \begin{array}{ll}
        \mathcal{R}_{0_{\rm BZ}}^{\rm out} &= \displaystyle \mathcal{R}_{0_B}^{\rm out}+\mathcal{R}_{0_Z}^{\rm out}\\[1.5ex]
         &= \displaystyle\frac{\xi N_0}{(\gamma+\delta)\Big(d_b+c\sigma \frac{Z^*}{h_m}\Big)}\Bigg[\frac{\beta}{h_b} + \frac{\beta_z}{h_z}\sigma\frac{Z^*}{d_zh_m}\Bigg]> 1
    \end{array}
    \label{outbreak_cond}
\end{eqnarray}
where \begin{eqnarray*}
\begin{array}{ll}
     \displaystyle \mathcal{R}_{0_B}^{\rm out}&=\displaystyle\frac{\xi N_0}{(\gamma+\delta)\Big(d_b+c\sigma \frac{Z^*}{h_m}\Big)}\frac{\beta}{h_b},\\[2ex]
     \displaystyle\mathcal{R}_{0_Z}^{\rm out}&=\displaystyle\frac{\xi N_0}{(\gamma+\delta)\Big(d_b+c\sigma \frac{Z^*}{h_m}\Big)}\frac{\beta_z}{h_z}\sigma\frac{Z^*}{d_zh_m}.
\end{array}
\end{eqnarray*}
Here, $\displaystyle \mathcal{R}_{0_{\rm BZ}}^{\rm out}$ denotes the basic reproduction number. Also, $\displaystyle \mathcal{R}_{0_B}^{\rm out}$ and $\displaystyle \mathcal{R}_{0_Z}^{\rm out}$ represent contributions associated with the free-living bacterial route and the bacteria-associated zooplankton route, respectively, to the initial outbreak growth. $N_0$ denotes the initial human population size.

\subsubsection{Impact of bacteria-zooplankton association on $\mathcal{R}_{0_{\rm BZ}}^{\rm out}$}
In the absence of the bacteria-zooplankton ($B$-$Z$) association (i.e. $\sigma = 0$), the condition for the initial outbreak growth $\displaystyle \mathcal{R}_0^{\rm out}= \mathcal{R}_{0_{\rm BZ}}^{\rm out}|_{\sigma = 0}=\displaystyle\frac{\xi N_0}{d_b(\gamma+\delta)}\frac{\beta}{h_b} >1$, aligns with the classic SIRB cholera model \citep{Codeco01}. The colonization of zooplankton by bacterial cells reduces the free-living bacterial density, thereby lowering $\displaystyle \mathcal{R}_{0_B}^{\rm out}$. This association simultaneously increases $\displaystyle \mathcal{R}_{0_Z}^{\rm out}$ through the transmission via commensal zooplankton. As a result, for a specific $\beta$, the $B$-$Z$ association can lead to an increase or decrease of $\displaystyle \mathcal{R}_{0_{\rm BZ}}^{\rm out}$ relative to $\displaystyle \mathcal{R}_0^{\rm out}$ depending on $\sigma$ and $\beta_z$. 

The two-parameter space ($\beta$-$\beta_z$) can be divided into six regions using the three lines, $\displaystyle \mathcal{R}_{0}^{\rm out}=1$ (dashed vertical line), $\displaystyle \mathcal{R}_{0_{\rm BZ}}^{\rm out}=1$ (red line), and $\displaystyle \mathcal{R}_{0_{\rm BZ}}^{\rm out}=\displaystyle \mathcal{R}_{0}^{\rm out}$ (black line) (Fig.~\ref{fig2-beta_vs_betaz}a). While, in region \boxed{1} and \boxed{6}, the $B$-$Z$ association increases $\displaystyle \mathcal{R}_{0_{\rm BZ}}^{\rm out}$ compared to $\mathcal{R}_0^{\rm out}$, in \boxed{2} and \boxed{3}, it decreases the same. Here, the decrease in $\displaystyle \mathcal{R}_{0_{\rm BZ}}^{\rm out}$ can be achieved with reduced $\beta_z$, which signifies the reduced exposure to zooplankton contaminated water. It is interesting to note that in \boxed{6}, an outbreak cannot initiate without the $B$-$Z$ association, highlighting the significance of zooplankton-driven route. This contrasts with the behavior in \boxed{3}, where an outbreak that can grow without this association decays in its presence. In regions \boxed{4} and \boxed{5}, the outbreak fails to grow irrespective of the presence or absence of the $B$-$Z$ association, making it insignificant in the context of our study. For a lower $\beta$, a relatively high $\beta_z$ can still trigger an outbreak when $\mathcal{R}_{0_{\rm BZ}}^{\rm out}>1$. This implies that zooplankton-mediated transmission can compensate for a low transmission rate via free-living bacteria.

For any point in \boxed{1}, e.g. $P_1$, an increased bacteria-zooplankton association rate ($\sigma$) always increases $\mathcal{R}_{0_{\rm BZ}}^{\rm out}$ (red line Fig.~\ref{fig2-beta_vs_betaz}(b)). Further, slope of the line $\displaystyle \mathcal{R}_{0_{\rm BZ}}^{\rm out}=1$ decreases with increasing $\sigma$ (red-dashed line in Fig.~\ref{fig2-beta_vs_betaz}(a)) (see \ref{app_out_sec2} for details). As a result, regions \boxed{6} and \boxed{3} expand at the expense of \boxed{5} and \boxed{2}, respectively. Consequence of this can be observed by following the fate of the outbreak on increasing $\sigma$ at points $P_4$ and $P_5$. At $P_4$, where the outbreak initially grows ($\mathcal{R}_{0_{\rm BZ}}^{\rm out}>1$), the system is driven to the region where the outbreak decays ($\displaystyle \mathcal{R}_{0_{\rm BZ}}^{\rm out}<1$) with increasing $\sigma$ (green-dashed line in Fig.~\ref{fig2-beta_vs_betaz}(b)). The converse is true for the point $P_5$ as shown by red-dashed line in Fig.~\ref{fig2-beta_vs_betaz}(b). Additionally, there are some points in \boxed{2}, e.g. $P_3$, for which increasing $\sigma$ will decrease $\mathcal{R}_{0_{\rm BZ}}^{\rm out}$ while still maintaining it above 1 (green line in Fig.~\ref{fig2-beta_vs_betaz}(b)). The black ($P_2$) line in Fig.~\ref{fig2-beta_vs_betaz}(b) represents the scenario $\displaystyle \mathcal{R}_{0_{\rm BZ}}^{\rm out} = \mathcal{R}_0^{\rm out}$ under any association rate ($\sigma$). Moreover, from Fig.~\ref{fig2-beta_vs_betaz}(b), it is important to note that for a fixed $\beta$, changes in $\beta_z$ (e.g., $P_1$-$P_4$) lead to relatively larger variations in $\mathcal{R}_{0_{\rm BZ}}^{\rm out}$ under a high $\sigma$ compared to a lower one. This can be explained by the fact that under favorable conditions, large number of \textit{Vibrio}-associated zooplankton cells significantly impact the disease spread.
% This can be explained by the fact that a high $\sigma$ leads to an increased density of $Z_B$, thus leading to a larger impact on $\mathcal{R}_{0_{\rm BZ}}^{\rm out}$. 
We use global sensitivity analysis to show that the above result holds true for all values of $\beta$ (see Fig.~\ref{figS1-prcc_R0} in \ref{sensitivity}).

Notably, in this study, we always consider $d_z<d_b$, which is well supported by ecological evidences \citep{Codeco01, kolaye2019mathematical, da2020non, hirst2002mortality,di2017non}. An increase in $d_z$ increases the slope of the $\mathcal{R}_{0_{\rm BZ}}^{\rm out}=\mathcal{R}_0^{\rm out}$ line, thereby reducing the region \boxed{1} (see Fig.~\ref{fig7-app_beta_vs_betaz_dz} in Appendix). This indicates a diminished effect of transmission via zooplankton as the persistence of bacterial cells in association with zooplankton is reduced. Also, the size of the regions \boxed{1}-\boxed{6} may vary for different human demographic factors in the endemic scenario (see Fig.~\ref{figS4-long_term_beta_vs_betaz} in \ref{app_long_term_region}).

\begin{figure}[h!]
\begin{center}
\includegraphics[width=0.99\textwidth]{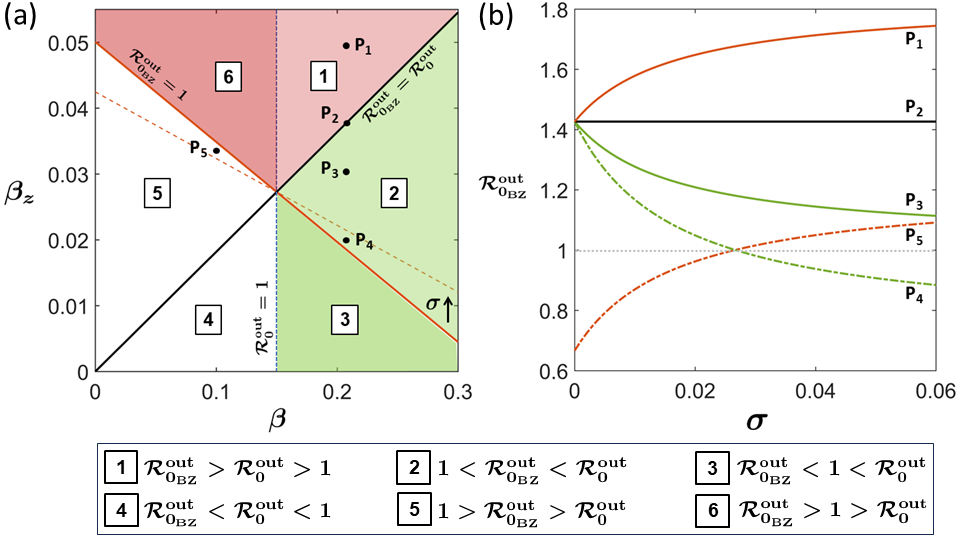}
\end{center}
\caption{The overall impact of the bacteria-zooplankton association on $\displaystyle \mathcal{R}_{0_{\rm BZ}}^{\rm out}$. (a) Effect of the transmission rates $\beta$, $\beta_z$. The red and black lines depict $\displaystyle \mathcal{R}_{0_{\rm BZ}}^{\rm out}=1$ and $\displaystyle \mathcal{R}_{0_{\rm BZ}}^{\rm out}=\mathcal{R}_0^{\rm out}$, respectively. The vertical dashed line indicates $\mathcal{R}_0^{\rm out}=1$, corresponding to the absence of the $B$-$Z$ association. An increase in $\sigma$ expands regions \boxed{6} and \boxed{3} at the expense of \boxed{5} and \boxed{2}, respectively, as indicated by the red-dashed line. (b) The effect of the association rate ($\sigma$) on points $P_1$-$P_5$, which belong to different regions demonstrated in (a).
 }
\label{fig2-beta_vs_betaz}
\end{figure}

\begin{figure}[h!]
\begin{center}
{\includegraphics[width=0.99\textwidth]{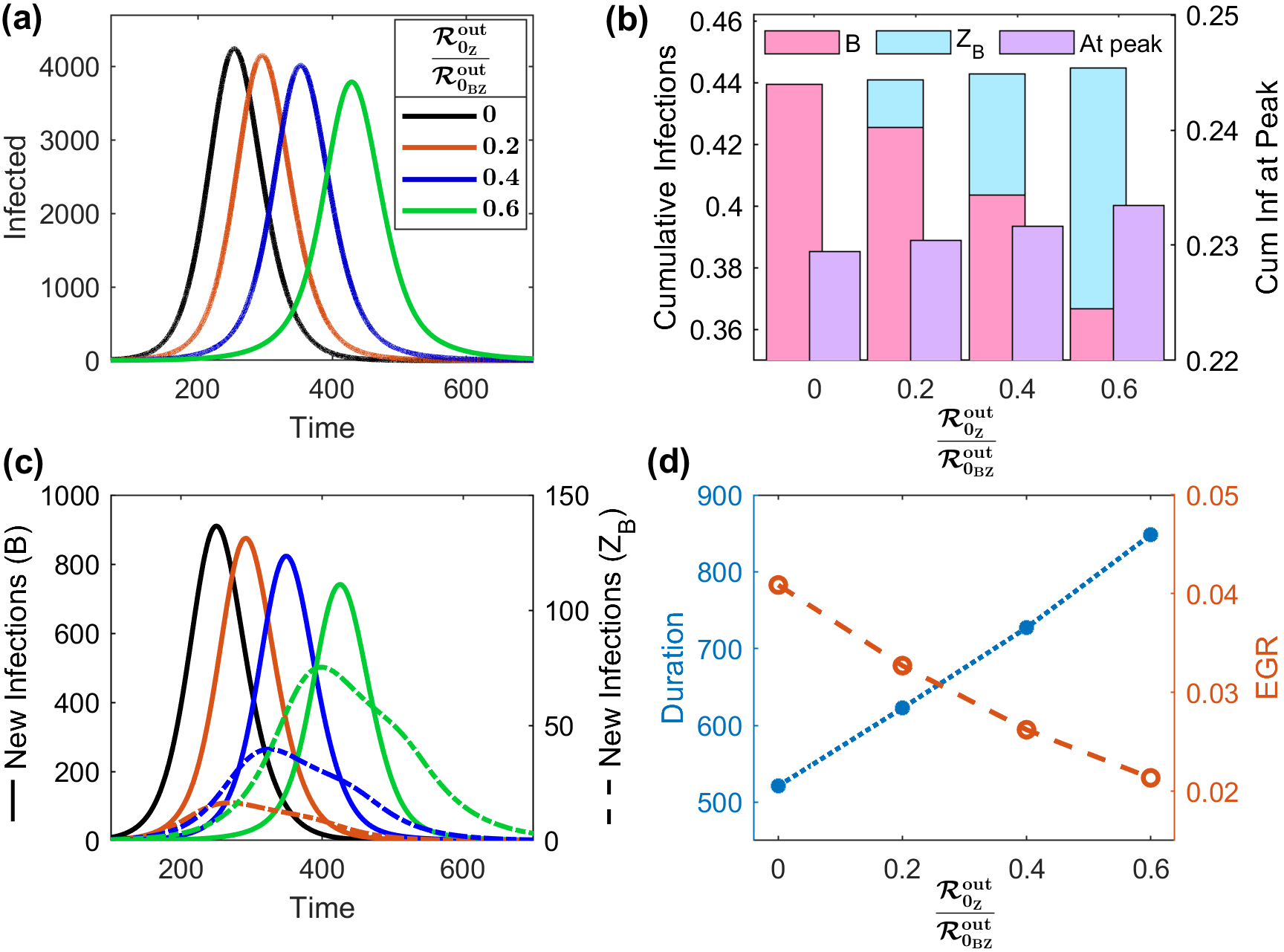}}
\end{center}
\caption{Effect of changing the relative contribution of transmission routes while keeping $\mathcal{R}_{0_{\rm BZ}}^{\rm out}$ fixed. (a) Active infections, (b) cumulative infections (proportion of the population), (c) new infections via the $B$ (solid lines) and $Z_B$ (dashed lines) route, and (d) epidemic duration (blue dotted) and initial epidemic growth rate (EGR) (red dashed). Increasing $\mathcal{R}_{0_Z}^{\rm out}$ decreases initial EGR, delays and lowers the peak, but increases cumulative infections both at the peak and at the end of the outbreak.
}
\label{fig3-rel_cont}
\end{figure}

\subsubsection{Relative contribution of transmission routes}
% Recalling the bacteria-carrying capability of zooplankton cells, it is important to note that even minimal exposure to zooplankton-contaminated water can significantly affect the overall transmission potential, comparable to higher contact with free bacteria-contaminated water (Fig.~\ref{fig2-beta_vs_betaz}(a) black line). 
The ratio of contributions to the basic reproduction number ($\mathcal{R}_{0_{\rm BZ}}^{\rm out}$) from zooplankton ($Z_B$) and bacterial ($B$) routes can be expressed as $\displaystyle \mathcal{R}_{0_Z}^{\rm out} /\mathcal{R}_{0_B}^{\rm out} = (\sigma Z^* h_b \beta_z)/(d_Z h_z h_m \beta)$. 
% Notably, in the case when contact with zooplankton is as frequent as contact with free-living \textit{Vibrios} ($\beta \sim \beta_z$), the $Z_B$ route contributes approximately 1.5 times to $\mathcal{R}_{0_{\rm BZ}}^{\rm out}$ route than the $B$ route with the association rate $\sigma=0.02$. 
Environmental factors, such as fluctuations in the temperature of the coastal sea surface, salinity, pH, heatwave, rainfall, floods, and ecological events such as plankton blooms, can influence the bacteria-zooplankton association \citep{huq1984influence,shackleton2024mechanisms,conner2016staying}, causing one transmission route to become more prominent over the other. Therefore, it is crucial to investigate for a fixed $\mathcal{R}_{0_{\rm BZ}}^{\rm out}$, how altering relative contributions ($\mathcal{R}_{0_B}^{\rm out}$, $\mathcal{R}_{0_Z}^{\rm out}$) through different transmission routes can impact disease progression and overall disease burden. To this aim, we keep $\mathcal{R}_{0_{\rm BZ}}^{\rm out}=\mathcal{R}_0^{\rm out}>1$, which always implies $\beta_z=\displaystyle \frac{ch_zd_z}{h_bd_b}\beta$ (black line in Fig.~\ref{fig2-beta_vs_betaz}(a)). For a fixed $\beta$, corresponding to each pair of relative contributions ($\mathcal{R}_{0_B}^{\rm out}$, $\mathcal{R}_{0_Z}^{\rm out}$), the unique association rate $\sigma$ is given by (see details in \ref{app_rel_cont})
\begin{equation*}
    \displaystyle \sigma= \displaystyle \frac{d_bh_m}{cZ^*} \Bigg(\frac{\mathcal{R}_{0_{\rm BZ}}^{\rm out}}{\mathcal{R}_{0_B}^{\rm out}}-1\Bigg)=\displaystyle \frac{d_bh_m}{cZ^*} \frac{\mathcal{R}_{0_Z}^{\rm out}}{\mathcal{R}_{0_B}^{\rm out}}.
\end{equation*}
Now, we compare outbreak trajectories by varying $\sigma$ that captures scenarios for different relative contributions from the both routes. We examine various characteristics of these trajectories such as peak values, peak timing, cumulative infections, and epidemic duration (see Fig.~\ref{fig3-rel_cont}).

% Note that to keep $\mathcal{R}_{0_{\rm BZ}}^{\rm out}=\mathcal{R}_0^{\rm out}$, we always have $\beta_z=\displaystyle \frac{ch_zd_z}{h_bd_b}\beta$ (black line in Fig-2(a)).

In the absence of \textit{Vibrio}-zooplankton association ($\sigma=0$), transmission only via the bacterial route drives the prevalence with a relatively higher and earlier peak, followed by a steeper decline (black line in Fig.~\ref{fig3-rel_cont}(a)). When the $Z_B$ route contributes a moderate to high proportion ($20\%$ to $60\%$), the outbreak progresses over a longer duration with delayed and reduced epidemic peak (Fig.~\ref{fig3-rel_cont}(a),(d)). However, it results in a slightly increased proportion of cumulative infections both at peak and at the end of the outbreak, compared to the $\sigma=0$ scenario (Fig.~\ref{fig3-rel_cont}(b)). The delayed dynamics in the presence of the $B$-$Z$ association arise because the contribution from the $Z_B$ route leads to a comparatively slower initial epidemic growth rate (EGR) (Fig.~\ref{fig3-rel_cont}(d) and \ref{EGR} for EGR calculations). The zooplankton-mediated transmission route can be viewed as a delayed transmission pathway, as the zooplankton have to first be colonized by bacteria cells before transmitting the infection to susceptible persons. The increase in cumulative infections at the peak in the presence of the zooplankton route means greater depletion of susceptible population before the peak and indicates an increased threshold of herd immunity for the same $\displaystyle \mathcal{R}_{0_{\rm BZ}}^{\rm out}$ (Fig.~\ref{fig3-rel_cont}(b)). These observations emphasize that predictions about the outbreak trajectories and disease severity can not be precisely determined only by analyzing $\mathcal{R}_{0_{\rm BZ}}^{\rm out}$, the contribution of transmission routes is also important.

In spite of the increasing cumulative infections, the $Z_B$ route produces fewer new infections compared to the $B$ route in each case (Fig.~\ref{fig3-rel_cont}(b)(c)). In fact, the epidemic peak's size and timing are primarily determined by the $B$ route, which is responsible for the majority of new infections (Fig.~\ref{fig3-rel_cont}(c)). Increasing $\mathcal{R}_{0_{Z}}^{\rm out}$ leads to reduced density of $B$ cells which results in a lower peak of infections. This occurs because, when $\displaystyle \mathcal{R}_{0_{\rm BZ}}^{\rm out}=\mathcal{R}_0^{\rm out}$, (i.e. $\beta_z=\displaystyle \frac{ch_zd_z}{h_bd_b}\beta$, black line in Fig.~\ref{fig2-beta_vs_betaz}(a)), we always have $\beta_z<\beta$ as $d_z<d_b$ and thus, the $Z_B$ route always provides a lower force of infection than the $B$ route.

\subsubsection{Nature of outbreak trajectories}
In this section, we investigate the nature of the outbreak trajectories within region \boxed{1} of Fig.~\ref{fig2-beta_vs_betaz}(a) while varying the bacteria-zooplankton association rate ($\sigma$) for different zooplankton-mediated transmission rate ($\beta_z$). For high $\beta_z = 0.08$, an increased $\sigma$ leads to both earlier (EGR increases) and higher peaks compared to $\sigma=0$ case (Fig.~\ref{fig4-sigma_vs_all}(b),(c),(e) and Fig.~\ref{FigS2-sigma_vs_EGR} in Appendix). When $\beta_z$ is low ($\beta_z=0.04$), a delayed and lower peak is observed. However, for intermediate $\beta_z$ (=0.05), while the peak is delayed, the number of infections at peak increases (Fig.~\ref{fig4-sigma_vs_all}(b-c),(f-g)). The decrease in EGR with increasing $\sigma$ for both $\beta_z = 0.04$ and $\beta_z=0.05$ explains the delayed dynamics in these scenarios (see Fig.~\ref{FigS2-sigma_vs_EGR} in Appendix). Since an outbreak with a delayed and lower peak can still result in a larger outbreak size (due to higher $\mathcal{R}_{0_{\rm BZ}}^{\rm out}$ at increased $\sigma$, Fig.~\ref{fig4-sigma_vs_all}(a)), it may not be easy to predict the outbreak size based on peak value and peak timing when zooplankton-mediated transmission is involved (Fig.~\ref{fig4-sigma_vs_all}(g)). For all of the above scenarios, the epidemic duration increases noticeably, compared to the $\sigma=0$ case (Fig.~\ref{fig4-sigma_vs_all}(d)).

Under fixed $\sigma=0.03$, an increased $\beta_z$ leads to an increased duration of epidemic in spite of earlier and larger peak value (Fig.~\ref{fig4-sigma_vs_all} and Fig.~\ref{figS3-time_series} in Appendix). This is due to the fact that increasing $\beta_z$ results in a considerably slower asymptotic convergence of the infection trajectory to the disease-free state after the peak infections, thus prolonging the duration of the outbreak (for example, see Fig.~\ref{fig4-sigma_vs_all}(e)). This is in contrary to the usual notion whereby an epidemic trajectory with a higher reproduction number should results in a relatively larger and faster peak followed by a quicker decline. This observation underscores the importance of zooplankton in promoting pathogen persistence during inter-epidemic periods by serving as a \textit{V. cholerae} reservoir. The maintenance of lower-level infections for a long period continues to influence the number of infections after the peak or equivalently after achieving the herd immunity threshold. In this context, the next subsection explores the epidemic overshoot phenomenon, which accounts for outbreak severity in the post-peak phase.

\begin{figure}[H]
\begin{center}
\includegraphics[width=0.99\textwidth]{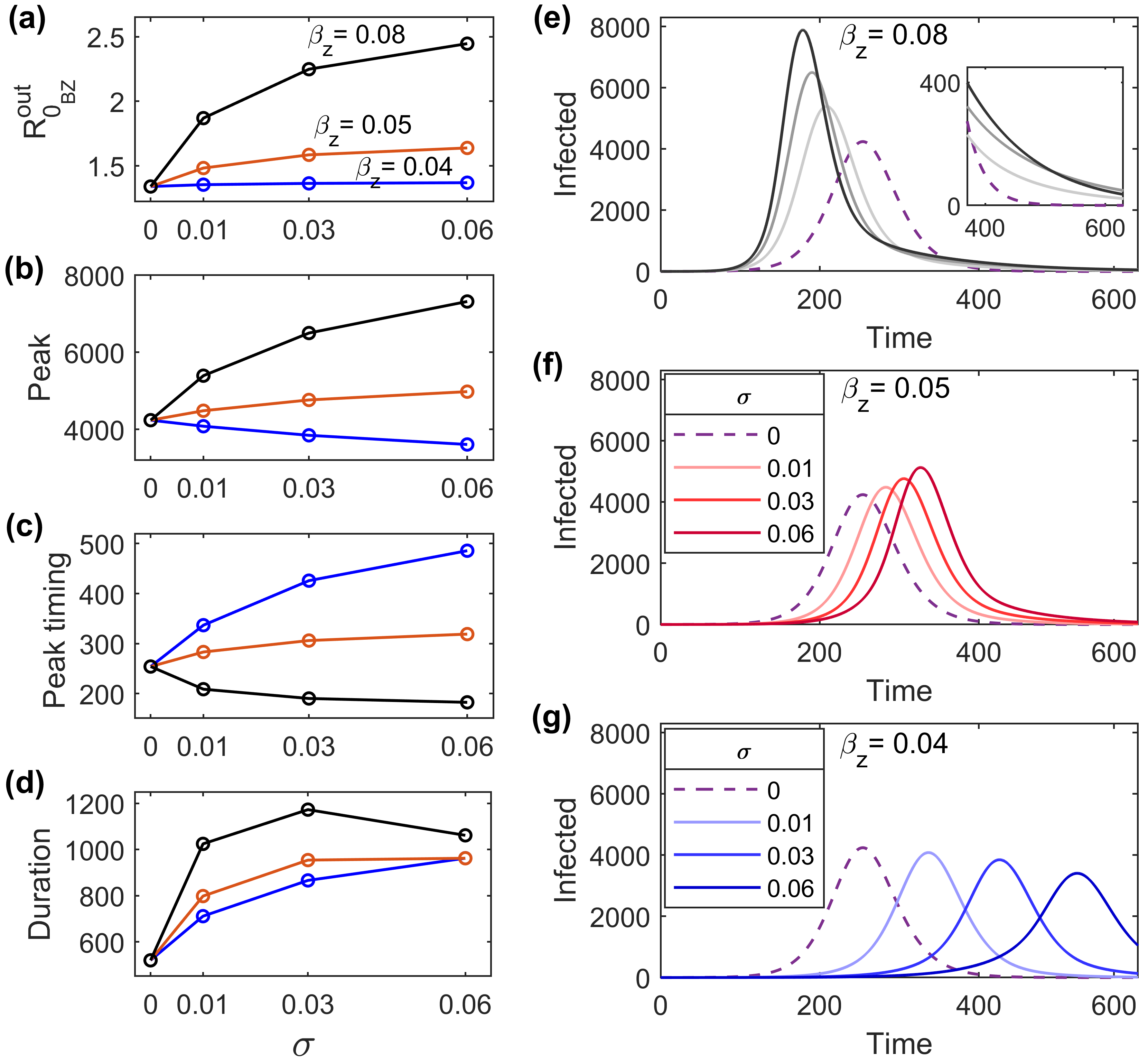}
\end{center}
\caption{$\mathcal{R}_{0_{\rm BZ}}^{\rm out}$ (a), peak infections (b), peak timing (c), epidemic duration (d) and the outbreak trajectories (e-g) under varying bacteria-zooplankton association rate ($\sigma$) for different transmission rates via zooplankton ($\beta_z$) within region \boxed{1} of Fig.~\ref{fig2-beta_vs_betaz}(a). 
% Under a fixed $\sigma$, increasing $\beta_z$ shortens the time to peak infection and increases the peak size, but elevates lower-level maintenance of infections over an extended post-peak phase before finally dying out (Fig.~S3 in supplementary material).}
% Although both $\displaystyle \mathcal{R}_0^{\rm out}$ and the final size increases with $\sigma$ for each $\beta_z>\beta_z^c$, the outbreak trajectories exhibit variations in peak value and peak timing. The reduced and increased EGR with respect to $\sigma$ in case of lower and higher $\beta_z$, respectively, explain the delayed and rapid dynamics of an outbreak (see Fig.~\ref{EGR-fig} in Appendix~\ref{EGR}).
}
\label{fig4-sigma_vs_all}
\end{figure}

\begin{figure}[H]
\begin{center}
{\includegraphics[width=0.99\textwidth]{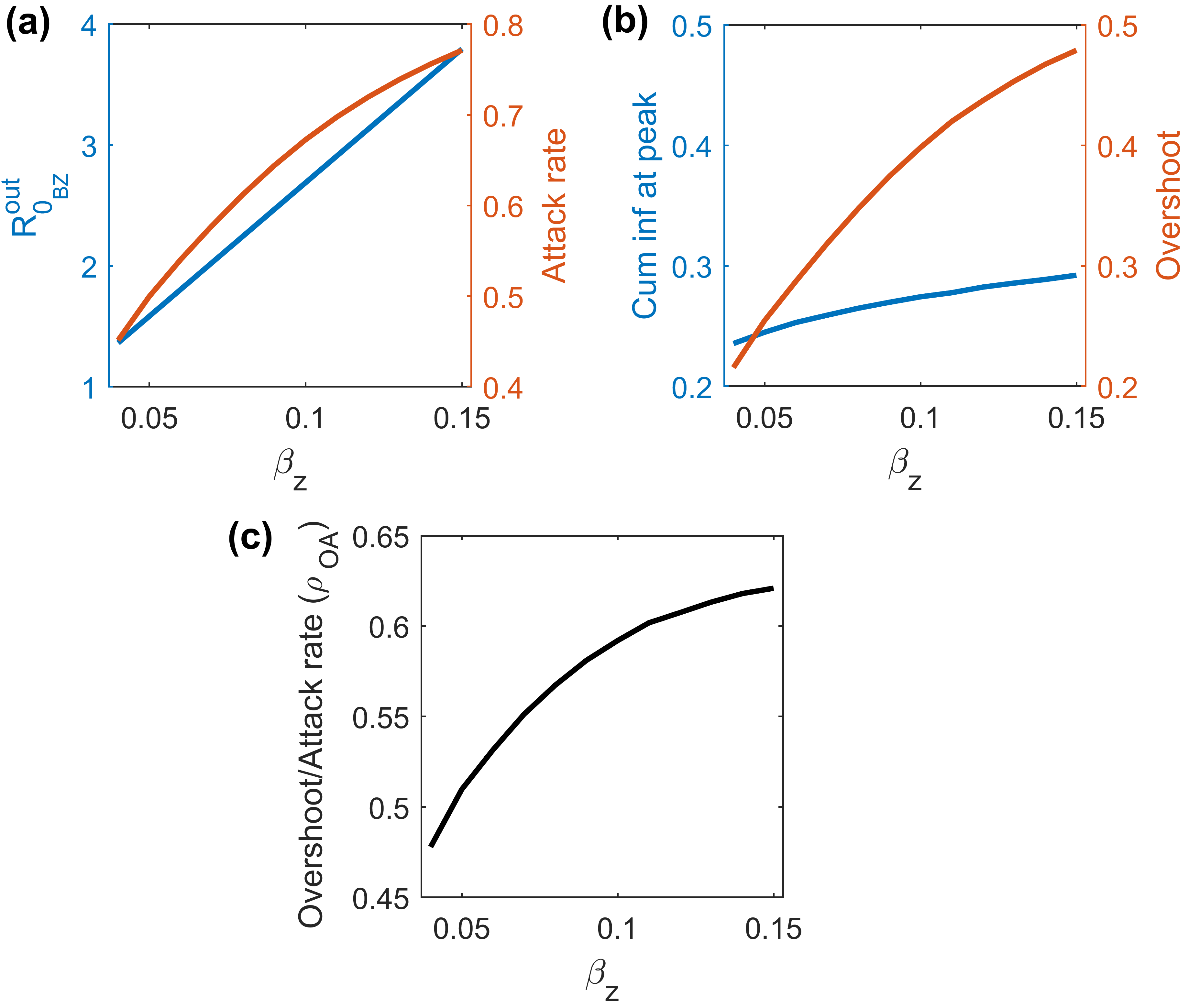}}
\end{center}
\caption{Effect of zooplankton-mediated transmission ($\beta_z$) on (a) $\mathcal{R}_{0_{\rm BZ}}^{\rm out}$ and attack rate, (b) cumulative infections at peak (as proportion of the total population) and overshoot, and (c) ratio of overshoot to attack rate ($\rho_{\rm OA}$). Parameter: $\sigma=0.03$.
% Both $\mathcal{R}_0^{\rm out}$ and attack rate increases with $\beta_z$. Although $\beta_z$ increases both the cumulative infections at peak and the overshoot, the rise in the latter outweighs the increase in the former. Consequently, $\rho_{\rm OA}$ increases over wide range of $\beta_z$. 
% \sout{While peak time consistently decreases with increasing $\beta_z$ (blue line), epidemic duration behaves non-monotonically (red line).}
}
\label{fig5-betaz_vs_overshoot}
\end{figure}

\subsubsection{Epidemic overshoot}

Epidemic overshoot refers to the proportion of the population that becomes infected after the peak of infection has passed, i.e. during the post-peak phase. It is equivalent to the difference between the herd immunity threshold and the attack rate (i.e. proportion of the infected population) \citep{handel2007best}. Notably, the ratio of overshoot to attack rate ($\rho_{\rm OA}$) is an important metric for quantifying the fraction of total infections occurring at the overshoot phase \citep{nguyen2023fundamental}. In the absence of bacteria-zooplankton association ($\sigma=0$), in line with a simple SIR system, higher transmission rates from the free-living bacterial route ($\beta$) always burn a significant portion of the population before reaching peak prevalence, leading to a sharp rise in infections and leaving fewer susceptible individuals for the overshoot phase. As a result, there always exists an upper bound on the overshoot with respect to $\beta$, and also the ratio of overshoot to attack rate ($\rho_{\rm OA}$) decreases monotonically as $\beta$ increases in case of $\sigma=0$ (see Fig.~\ref{fig8-app_beta_vs_overshoot} in Appendix).

In presence of bacteria-zooplankton association ($\sigma\neq 0$), increasing transmission via zooplankton ($\beta_z$) for a fixed $\beta$ alters the above observations. We find that, along with cumulative infections at the peak and the attack rate, the overshoot also increases substantially across a wide range of $\beta_z$ (see Fig.~\ref{fig5-betaz_vs_overshoot}(a),(b)). Interestingly, it appears that the increase in the overshoot due to $\beta_z$ profoundly surpasses the rise in cumulative infections at the peak (solid red and blue line Fig.~\ref{fig5-betaz_vs_overshoot}(b)). For instance, $50\%$ increase in $\beta_z=0.05$ results in $7\%$ and $30\%$ increase in cumulative infections at peak and the overshoot, respectively (Fig.~\ref{fig5-betaz_vs_overshoot}(b)). Consequently, the ratio of overshoot to attack rate ($\rho_{\rm OA}$) increases significantly over a wide range of $\beta_z$, which is contrary to what is usually observed in SIR systems (see Fig.~\ref{fig5-betaz_vs_overshoot}(c)). This result is particularly interesting, indicating a substantial overshoot in cases where transmission via zooplankton is more pronounced. Therefore, alongside transmission from the free-living bacterial route, even moderate transmission via zooplankton infects a larger portion of the susceptible population in the post-peak phase by sustaining lower-level infections over an extended period before eventually fading out. This underscores that maintenance of control measures during the post-peak phase becomes crucial to prevent additional overshoot. It should be noted that extremely large $\beta_z$ values (for which $\mathcal{R}_{0_{\rm BZ}}^{\rm out}$ may not be feasible for cholera (Eq.~\eqref{outbreak_cond})) cause both the overshoot and $\rho_{\rm OA}$ to decline again (not shown here).

\subsection{Disease management via water filtration}\label{sec:result_filtration}

In the Bengal delta region, numerous studies have documented seasonal fluctuations in \textit{Vibrio} cells attached to zooplankton, with the highest concentrations observed during the early spring and summer, coinciding with the peak zooplankton populations \citep{de2011role, huq2005critical, jutla2012satellite, shackleton2024mechanisms}. We note that our model, by explicitly considering the various ecological time scales involved, presents an opportunity to incorporate seasonality and its connection to disease dynamics and potential control measures. To qualitatively mimic the emergence of cholera infections preceded by biannual plankton blooms, we consider the phytoplankton carrying capacity ($K$) to be a periodic function given by  
\begin{equation*}
    K(t)=K_0(1+d \sin{(2\pi t/p)}).
\end{equation*}
Here $K_0$ stands for the baseline phytoplankton carrying capacity, while $d$ and $p$ denote the amplitude and period of oscillation, respectively. This assumption is reasonable as the carrying capacity depends on fluctuations in climatic and environmental factors such as coastal sea surface temperature, nutrient load, salinity, pH, flooding and streamflow \citep{shackleton2024mechanisms, freund2006bloom, d2012algal, baracchini2017seasonality, Codeco01, righetto2012role, kolaye2019mathematical}. In this case, we also consider human demographic and immunity factors ($\Lambda, \mu, \omega \ne 0$).

%Note that we make sure that $K(t)$ remains below the carrying capacity beyond which the phytoplaremains within the range the stable coexistence range, {\color{red}given by Eq.~S3 in supplementary material. Note that $K(t)$ always remains below $K^c$, given by Eq.~S3 in supplementary material. 

% The copepod-associated \textit{Vibrio cholerae} cells appear to be effectively protected from the effects of alum, a chemical commonly recommended for purifying natural waters used for drinking \citep{chowdhury1997effect}. Even after subsequent chlorination, nearly $90\%$ of bacteria cells may remain viable, posing a potential threat to human health. 

The copepod-associated \textit{Vibrio cholerae} cells remains viable even after treatment of household water with chemical disinfectants like alum and chlorination \citep{chowdhury1997effect}. However, a sari cloth folded four times can remove up to $99\%$ bacteria-associated zooplankton cells, as demonstrated in the laboratory-based study by \cite{huq1996simple}. Furthermore, an experimental study conducted in Matlab, Bangladesh, by \cite{colwell2003reduction} found that filtration using simple sari and nylon cloths can remove around $90\%$ of bacteria cells attached to copepods. Hence simple filtration of household water can be a more effective method of controlling cholera compared to chemical-based purification. This is particularly relevant in adverse situations including humanitarian crises and climate-driven extreme weather events, where access to clean water is limited, making an inexpensive, easy-to-use, and socially acceptable household water treatment method like filtration essential \citep{venkatesan2024new, huq2010simple}. 
% Further, adverse situations including humanitarian crises and climate-driven extreme weather events that limit access to clean water, may demand an inexpensive, easy-to-use, and socially acceptable household water treatment measure like filtration \citep{venkatesan2024new, huq2010simple}.
Additionally, the waning of vaccine-induced immunity, combined with the ongoing critical shortage of oral vaccines (OCV) further reinforces the need for household control measures \citep{cholera12, enserink2025cradle}.

Filtration is a point-of-use (POU) water purification method, typically applied when collecting water from reservoirs such as lakes and ponds \citep{sobsey2008point, taylor2015impact, colwell2003reduction}. By removing planktonic organisms, filtration effectively reduces zooplankton concentration and thus the infectious dose in household water \citep{fung2014cholera, sobsey2008point}. This, in turn, directly impacts the force of infection via the zooplankton-driven transmission route. Other POU measures, such as boiling and household chlorination \citep{fewtrell2005water}, are not considered here, as our focus is on assessing the impact of controlling infection through the zooplankton-mediated route. To model filtration, we modified the zooplankton-driven force of infection term as follows:
\begin{equation}
\displaystyle FI_Z(t)=\begin{cases}
\displaystyle\beta_z \frac{(1-e_f)Z_B(t)}{h_z+(1-e_f)Z_B(t)}, \qquad &\text{during filtration}\\\\
\displaystyle\beta_z \frac{Z_B(t)}{h_z+Z_B(t)} \qquad\qquad &\text{else.}
     \end{cases}
     \label{filtration}
\end{equation}
Here $e_f$ stands as the filtration efficacy, which quantifies the proportion of dispelled zooplankton cells and depends on the procedure's accuracy and the mesh size of the used material. For instance, sari cloths appeared to be more effective than nylon nets in removing copepods \citep{colwell2003reduction}. Note that we assume an instantaneous effect of filtration on $\displaystyle FI_Z(t)$, which is reasonable as the water is typically collected on a daily basis for household use.

An important question is when to initiate filtration and how long the control measure should last. Notably, plankton blooms typically serve as an early warning indicator for cholera outbreaks, with a lag of nearly 8 weeks between plankton blooms and the rise in cholera infections, as observed in the Bengal Delta region \citep{de2008environmental, huq2005critical}. So, one can think of two indicators for initiation of filtration: the occurrence of phytoplankton blooms, which are mostly recognizable in water reservoirs, and an increase in cholera infections. Since plankton blooms generally last for about three months, we focus on implementing periodic filtration with same duration. We assess the outcome of filtration by measuring the reduction in the number of infections over a year, compared to the situation without filtration. Implementing filtration measures based on the two indicators implies different initiation times, which may lead to differences in cholera case reduction. We also analyze the impact of filtration initiated at all time points, while maintaining the same duration as above for each case.

\begin{figure}[h!]
    \centering
    \includegraphics[width=0.99\textwidth]{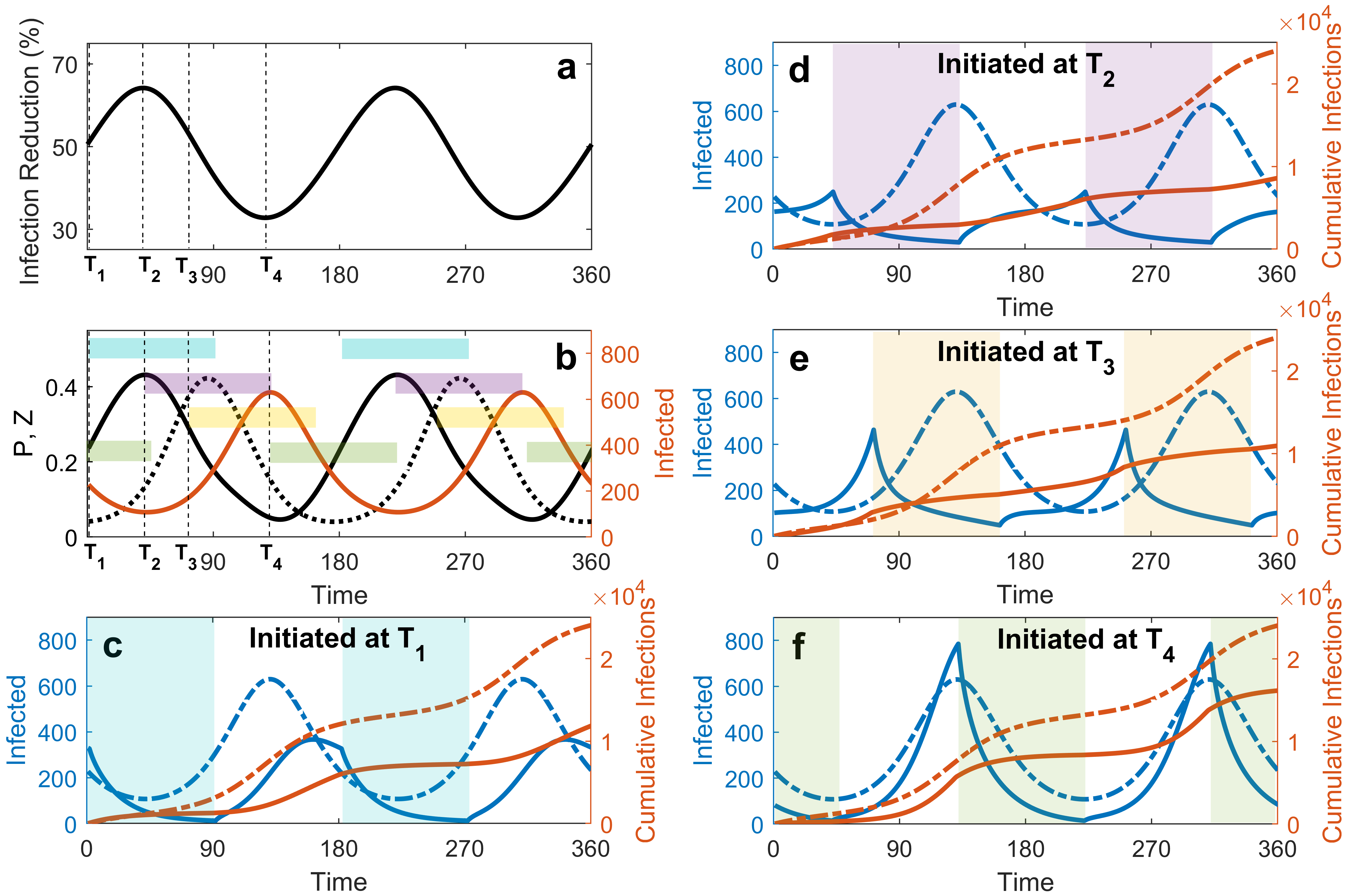}
    \caption{ Effect of 90-day periodic filtration on reducing cholera infections typically driven by seasonal biannual plankton blooms. (a) shows the percentage of infection reduction (solid line) for different initiation timings. Here $T_1$-$T_4$ represents four distinct initiation timings. (b) Illustrates potential indicators for initiating filtration, based on phytoplankton density (solid black line) and subsequent increases in infections (red line) along with the zooplankton abundance (dotted black line). The scenario when filtration is initiated at $T_1$ following the phytoplankton bloom (c), at $T_2$ when zooplankton abundance is increasing (d) (best-case scenario), at $T_3$ in response to an increasing trend in infections (e) and at $T_4$ when zooplankton abundance is decreasing (f) (worst-case scenario).
    % (c) shows the situation when filtration is initiated at $T_1$ following the phytoplankton bloom. (e) depicts the scenario when filtration is initiated at $T_3$ in response to an increasing trend in infections. (d) and (f) illustrate the best and worst outcome scenarios from filtration initiated at $T_2$ and $T_4$, respectively. 
    In (c)-(f), active (blue) and cumulative (red) infections are shown for both filtered (solid lines) and unfiltered (dashed lines) scenarios. The Shaded regions indicate the filtration periods. Here we consider filtration with $80\%$ efficiency (i.e. $e_f=0.8$), and the amplitude of oscillation $d=0.8$ of phytoplankton carrying capacity.
    }
    \label{fig6-filtration}
\end{figure}

The overall effect of filtration on reducing cholera infection at different initiation timings is illustrated in Fig.~\ref{fig6-filtration}. The solid line in Fig.~\ref{fig6-filtration}(a) indicates the percentage reduction in cholera infections associated with different filtration initiation timings (dashed vertical lines). Figure~\ref{fig6-filtration}(b) depicts two filtration indicators: the phytoplankton density (solid black line) and the subsequent rise in infections (red line), alongside zooplankton abundance (dotted black line). Four different filtration initiation timings, $T_1$-$T_4$ are considered among which $T_1$ and $T_3$ follow the former and the latter indicator, respectively (Fig.~\ref{fig6-filtration}(a)(b)). Additionally, the active and cumulative infections under filtration corresponding to timings $T_1$-$T_4$, along with the unfiltered scenario, are shown in Fig.~\ref{fig6-filtration}(c)-(f), respectively. We observe that, when filtration is employed with an efficacy of $80\%$ (i.e. $e_f = 0.8$), the reduction in cholera infections varies from approximately $32\%$ to $65\%$, depending on the timing of initiation (Fig.~\ref{fig6-filtration}(a)). Filtration initiated at $T_1$ following a phytoplankton bloom reduces around $50\%$ infections (Fig.~\ref{fig6-filtration}(c)), which is slightly less effective than the $54\%$ reduction observed when filtration is initiated at $T_3$ in response to increasing cholera infections (Fig.~\ref{fig6-filtration}(e)). 

The above results can be explained by tracking the formation and ingestion of $Z_B$ in our simulations.
% Here, $Z_B$ formation depends on the available concentration of $B$ cells during $Z$ abundance. 
Filtration initiated at $T_1$ ends when $Z$ abundance is close to its peak, leading to considerable formation of  $Z_B$ which are then ingested (Fig.~\ref{fig6-filtration}(c), Fig.~\ref{fig9-app_filtration}(a) in Appendix). On the other hand, when filtration is initiated at $T_3$, in spite of significant $Z_B$ formation, it can lead to relatively less ingestion during the high abundance period (Fig.~\ref{fig6-filtration}(e), Fig.~\ref{fig9-app_filtration}(c) in Appendix). Filtration initiated at $T_2$, which continues during periods of high $Z$ abundance, effectively restricts both formation and ingestion of $Z_B$, appears to be the best-case scenario (Fig.~\ref{fig6-filtration}(d), Fig.~\ref{fig9-app_filtration}(b) in Appendix). Hence, in this case, infections remain relatively low throughout the year (nearly $65\%$ reduction in Fig.~\ref{fig6-filtration}(d)). However, it is worth noting that, due to the difficulty in tracking zooplankton concentrations in the water column, this may not be a very feasible recommendation. The worst-case scenario arises from delayed filtration initiated near the infection peak at $T_4$, which fails to restrict both formation and ingestion of $Z_B$ (Fig.~\ref{fig9-app_filtration}(d) in Appendix). In this case, a rapid peak, even larger than the uncontrolled one arises although with a reduced total infection over the year (nearly $32\%$ reduction in Fig.~\ref{fig6-filtration}(f)). Overall, we observe that the timing of implementing filtration can play a major role in reduction of cholera infections over the year.

% filtration, when implemented timely and appropriately with a higher efficacy (e.g., $80\%$ as shown in Fig.~\ref{fig6-filtration}), significantly reduces cholera infections over the year.

\section{Discussion}

Cholera remains a reemerging disease of poverty for nearly 2 billion people in over 50 countries who have inadequate access to safe water and poor sanitation infrastructure \citep{amisu2024cholera}. 
% Increased climate variability and extreme weather events that disrupt sanitation and water access further exacerbate the issue \citep{venkatesan2024new}. Additionally, various environmental factors influencing the ecology of \textit{Vibrio cholerae} underscore the need for continued research on cholera. 
Despite the well-documented correlation between plankton blooms and cholera outbreaks over the past decades, the impact of plankton ecology on cholera dynamics has not yet been fully explored from a mathematical modeling perspective. To this aim, we integrate phytoplankton-zooplankton interactions into the classical human-bacteria (SIRB) cholera model through the ecological commensal association between \textit{V. cholerae} and zooplankton. This could lead different dynamics of cholera outbreak which has remain unexplored until now. For instance, an outbreak that might initially decay without the bacteria-zooplankton association can grow when the transmission via zooplankton is involved. On the other hand, this association can also lead to decline in initial infections which would have grown otherwise. Additionally, there will be scenarios where transmission via zooplankton increases the spread of the disease, i.e., $\mathcal{R}_{0_{\rm BZ}}^{\rm out}$ (Fig.~\ref{fig2-beta_vs_betaz}). 

The basic reproduction number ($\mathcal{R}_{0_{\rm BZ}}^{\rm out}$) alone may not provide complete picture about outbreak trajectories or disease severity. The relative contribution of transmission routes, influenced by ecological, environmental, and climatic factors, plays a significant role in determining the progression and impact of an outbreak (Fig.~\ref{fig3-rel_cont}). Although the zooplankton route is responsible for fewer infections compared to the free-living bacterial route, the outbreak persists longer in the presence of bacteria-associated zooplankton cells. The colonization of zooplankton by \textit{V. cholerae} cells prior to human exposure makes the zooplankton route a delayed transmission pathway. This is consistent with the the results of \cite{tien2010multiple}. For a fixed reproduction number, the dominant free-living bacterial route drives rapid outbreak growth with a higher peak. In contrast, even moderate transmission via the zooplankton route results in a comparatively slower outbreak progression with a reduced peak. However, this leads to a higher herd immunity threshold, larger final size, and longer outbreak duration. This finding aligns with the contrasting epidemiological patterns of cholera: the Ganges delta region experiences consistent, longer outbreaks involving reservoir transmission, while the African region often reports shorter, sporadic outbreaks without reservoir transmission, as noted by \cite{sack2021contrasting}. This result emphasizes the consideration of relative contributions of transmission routes while designing interventions. Also, it underscores the importance of maintaining control measures despite a slower initial growth of the outbreak, particularly in regions with evidence of \textit{V. cholerae} reservoirs. 

Even when the basic reproduction number ($\mathcal{R}_{0_{\rm BZ}}^{\rm out}$) increases due to zooplankton-mediated transmission, we can observe both an early or a delayed peak, depending on the different rates of the bacteria-zooplankton association. While the former is the usual expectation, the latter is unintuitive. However, in both cases, a larger outbreak size is reached due to prolonged infections at lower levels during the post-peak phase (Fig.~\ref{fig4-sigma_vs_all}). This highlights the role of environmental reservoirs in prompting the inter-epidemic persistence of pathogens, as noted in previous studies \citep{sack2021contrasting, lutz2013environmental, vezzulli2010environmental}. In fact, the increase in overshoot may surpass the rise in cumulative infections at the peak when the transmission rate via zooplankton increases (Fig.~\ref{fig5-betaz_vs_overshoot}). This finding highlights the possibility of substantial overshoot in regions with the evidence of \textit{V. cholerae} reservoirs, such as the Ganges delta region, and underscores the need to maintain control measures during the post-peak period. Relaxing control measures beyond the peak of infections, even after achieving herd immunity, could potentially lead to additional overshoot.

Lastly, we also study disease management via practical control measure such as water filtration, which has previously been deemed effective to reduce cholera in endemic regions like Matlab, Bangladesh \citep{colwell2003reduction}. Our findings indicate that the timing of implementing such a control measure can be a key to a substantial (approximately $32-65\%$) reduction in cholera infections (Fig.~\ref{fig6-filtration}). These observations highlight that a deeper understanding of the \textit{Vibrio}-plankton ecology could aid in developing signaling methods for effective filtration in regions endemic to cholera. Although the presence of environmental reservoirs of \textit{V. cholerae} in coastal ecosystems and their associated ecology complicate pathogen eradication, cholera endemicity can still be controlled with simple yet timely interventions.

Increased climate variability with ongoing global change and subsequent extreme weather events can influence pathogen ecology, thereby highlighting the need to integrate the same while studying disease dynamics. Our model takes a first step in this direction in the context of cholera and provides qualitative insights into transmission dynamics which could be useful to inform public health policies. Despite the usefulness of our model, it has few limitations. While some studies indicate that \textit{V. cholerae} can also associate with phytoplankton \citep{phyto2008seeligmann, islam2015role}, we choose to neglect this aspect in our study. This assumption can also be justified in line with the empirical studies \citep{turner2009plankton,huq2005critical}, where an increase in \textit{V. cholerae} concentration has reportedly corresponded to a significant decrease in phytoplankton abundance driven by higher nutrient availability and reduced level of antibacterial metabolites. Additionally, we do not take into account the rapid growth of \textit{V. cholerae} \citep{lutz2013environmental,huq1983ecological} and the increased pathogen virulence \citep{walther2004pathogen, vezzulli2010environmental} while attaching to zooplankton. Incorporating these ecological aspects in future cholera research could enhance our understanding of disease dynamics. The need of the hour is to shift the focus toward integrating the ecology of pathogens and their response to changing environment while predicting disease outbreaks.

\appendix
%%%%%%%%%%%%%%%%%%%%% Appendix-A %%%%%%%%%%%%%%%%%%%%%%%%%%%
\section{Parameter definitions and calibration}\label{param}
The model summarized by Eq.~\eqref{final_model} comprises a complex dynamical system involving organisms that vary across many orders of magnitude in size and have similarly diverse time scales associated with their population dynamics. It is important, therefore, to be precise in defining the values and units of all the model parameters under consideration (see Table~\ref{table-1}). As the likelihood of exposure to zooplankton-contaminated water is lower than that of bacteria-contaminated water, we assume the transmission rate via zooplankton ($\beta_z$) is less than or equal to the transmission rate via free-living bacteria ($\beta$). Moreover, we consider the mean dry weight of zooplankton to be approximately 2 $\mu g$ \citep{dumont1975dry}, with the smaller size copepods being more likely to pass through unfiltered water. A single copepod carries around $10^5$ \textit{V. cholerae} cells \citep{de2008environmental, colwell1992ecology, colwell1996viable}, resulting in a colonization coefficient, $c$, for the bacteria-zooplankton association of $10^5/0.002=5\times 10^7$ cells/(mg~dw). Furthermore, the free-living bacterial transmission reaches half of its maximum at the concentration $h_b=10^9$ \textit{V. cholerae} cells/L \citep{hartley2005hyperinfectivity}. So, the half-saturation constant for transmission via zooplankton, $h_z$, becomes $10^9/(5\times 10^7) = 20$ (mg~dw)/L. Notably, the parameter $h_z$ is inversely proportional to the bacterial colonization coefficient $c$. Additionally, in line with the empirical study \citep{lipp2002effects}, we assume that bacteria attach to zooplankton at a relatively lower concentration than $h_b$, with the maximum attachment rate occurring at a concentration that is 1/100th of $h_b$. This gives the half-saturation constant for $B$-$Z$ association $h_m=10^7$ cells/L.

% To calculate $h_z$, we assume a copepod may contain as many as $10^5$ \textit{Vibrio cholerae} cells \citep{de2008environmental, colwell1992ecology, colwell1996viable}. We also assume, $10^6$ \textit{V. cholerae} cells/ml is required so that the transmission rate is half of the maximum \citep{hartley2005hyperinfectivity} which is equivalent to 10 zooplankton/ml. If we further assume the mean dry weight of zooplankton to be approximately 1 $\mu g$ \citep{dumont1975dry} (the smaller size copepod has a larger chance of passing through unfiltered water), then the dry weight equivalent of the half-saturation constant is 10 $\mu g$ \textit{dry wt.} $ml^{-1}$ = 10 $mg$ \textit{dry wt.}$L^{-1}$ \citep{colwell1996viable}. 

%################################################
%#############################################
\section{Positive Invariance and boundedness}
% Follow Kolyae 2019, "Impacts of planktonic components on the dynamics of cholera epidemic: Implications from a mathematical model (Medda@MCS2024)" and "Plankton population and cholera disease transmission: A mathematical modeling study" for basic dynamics (Panja@IJBC2020), "Optimal Control Analysis of a Cholera Epidemic Model (Panja@BRL2019)". 
\begin{lemma}
The solutions (S(t),I(t),R(t),B(t),Z(t),P(t)) of the system~(2.1) are uniformly and ultimately bounded on\\ 
$\displaystyle\Omega=\Big\lbrace(S,I,R,B,Z_B, Z_F,P)\in \mathbb{R}_+^{7}| ~0\leq S,I,R\leq \frac{\Lambda}{\mu}, 0\leq B\leq \frac{\xi\Lambda}{\mu d_b}, 0\leq Z+\eta P \leq \frac{K\eta (r_p+\alpha)^2}{4r_p \alpha}\Big\rbrace$.
\end{lemma}
\begin{proof}
Model (2.1) can be expressed as the form 
\begin{eqnarray*}
	\frac{dX}{dt} & =A(X(t)), ~~X(0)=X_0\geq 0,
	\end{eqnarray*}
where $X=[S,I,R,B,Z_B,Z_F,P]^T$ and $ A(X(t))=(A_1(X),A_2(X),\dots,A_7(X))^T$.\\
Now, we have
\begin{align*}
	\frac{dS}{dt}|_{S=0}&=\Lambda + \omega R \geq 0, ~~\frac{dI}{dt}|_{I=0}=\frac{\beta SB}{h_b+B} + \frac{\beta SZ_B}{h_z+Z_B}\geq 0,\\[1.5ex] 
    \frac{dR}{dt}|_{R=0}&=\gamma I\geq 0,
~~\frac{dB}{dt}|_{B=0}=\xi I \geq 0, ~~\frac{dZ_B}{dt}|_{Z_B=0}= \sigma\frac{BZ_F}{h_m+B} \geq 0,\\[1.5ex]
\frac{dZ_F}{dt}|_{Z_F=0}&= \eta\frac{\alpha PZ_B}{h_p+P} \geq 0,	~~\frac{dP}{dt}|_{P=0}= 0.
	\end{align*}
Based on the lemma from \cite{yang1996permanence}, it follows that $\mathbb{R}_+^{7}$ is an invariant set. Consequently, any trajectory of system (2.1) that originates in $\mathbb{R}_{+}^7$ will remain within this domain for all time.

\noindent Let $N(t)=S(t)+I(t)+R(t)$ be the total human population at any given instant of time t. Adding first three equations of the system~(2.1) we get,
$\displaystyle \frac{dN}{dt}=\Lambda - \delta I -\mu N$. This implies that, $\displaystyle \frac{dN}{dt}\leq \Lambda -\mu N.$ Therefore by standard comparison theorem, there exists $t_1 \geq 0$ such that $\displaystyle N(t)\leq \frac{\Lambda}{\mu},$ for all $t\geq t_1$.\\

\noindent Now from the fourth equation of Eq.~(2.1), $\displaystyle \frac{dB}{dt}\leq \xi I - d_{b}B$. Again by comparison theorem, there exists $t_2 \geq 0$ such that $\displaystyle B(t)\leq \frac{\xi\Lambda}{\mu d_b}$ for all $t\geq t_2$.\\

\noindent To show the boundedness of the remaining three compartments, let $M=Z+\eta P=Z_B+Z_F+\eta P$. 

\noindent Then, $\displaystyle \frac{dM}{dt} = -d_zZ_B +\eta r_{p} P(1-\frac{P}{K}) - d_{z}Z_F$. 

\noindent Now let $\alpha>0$. Then $\displaystyle \frac{dM}{dt} +\alpha M = \eta r_{p} P(1-\frac{P}{K})+\alpha \eta P+ (\alpha - d_{z})Z_F +(\alpha -d_z) Z_B.$ We choose $\alpha$ such that $\alpha < d_z$. Then $\displaystyle \frac{dM}{dt} +\alpha M \leq {\rm max}\Big\{\eta r_{p} P(1-\frac{P}{K})+\alpha \eta P \Big\}=\frac{K\eta(r_p + \alpha)^2}{4 r_p}$. Again by comparison theorem, there exists $t_3 \geq 0$ such that $\displaystyle M(t)\leq \frac{K \eta(r_p + \alpha)^2}{4 r_p \alpha},$ for all $t\geq t_3$.
\end{proof}

%#################################################
%############################################
\section{The phytoplankton-zooplankton model}\label{appendix-phyto-zoo model}
Considering the total zooplankton density as $Z(t)=Z_F(t)+Z_B(t)$, from the last three equations of system~(2.1), we have the phytoplankton-zooplankton equation as
\begin{eqnarray}
\begin{array}{llll}
\displaystyle \frac{d Z}{dt} &=& \displaystyle \eta \frac{\alpha PZ}{h_p +P} -d_zZ,\\\\
\displaystyle \frac{d P}{dt} &=& \displaystyle r_pP(1-\frac{P}{K})-\frac{\alpha PZ}{h_p+P}.
\end{array}
\label{phyto_zoo_model}
\end{eqnarray}
Here, the independence of system~\eqref{phyto_zoo_model} from the human-bacteria (SIRB) dynamics reflects the commensal interaction between \textit{Vibrio cholerae} and zooplankton. The system (\ref{phyto_zoo_model}) has three equilibrium solutions: trivial $S_0=\displaystyle(0, 0)$, zooplankton-free $S_1=\displaystyle(0, K)$ and zooplankton-phytoplankton co-existence $S^*=\displaystyle(Z^*, P^*)$ equilibrium. Where the $P^*$ and $Z^*$ are given by the solution of the following two equations:
\begin{equation*}
    \displaystyle \eta \frac{\alpha P^*Z^*}{h_p +P^*} -d_zZ^*=0 ~~\text{and}~~
\displaystyle r_pP^*(1-\frac{P^*}{K})-\frac{\alpha P^*Z^*}{h_p+P^*}=0.
\end{equation*}
Solving the above equations we get,
\begin{eqnarray}
\begin{array}{llll}
\displaystyle P^*=\displaystyle \frac{d_zh_p}{\eta \alpha-d_z}, ~Z^*=\displaystyle \frac{r_p}{\alpha}(1-\frac{P^*}{K})(h_p+P^*)\\
\end{array}
\label{P_and_Z}
\end{eqnarray}
Here, $\displaystyle C_P =\frac{\eta\alpha}{d_z} > 1$ is the necessary condition for the existence of positive phytoplankton density. Also, $\displaystyle C_Z =\frac{KC_P}{(h_p+K)} > 1$ with $\displaystyle C_P > 1$ is the necessary condition for co-existence equilibrium density. Now, the local stability of the equilibrium solutions for the system (\ref{phyto_zoo_model}) follows the lemma below.

\begin{lemma}
%\label{pythagorean}
The trivial equilibrium solution $S_0= \displaystyle(0, 0)$ of the system (\ref{phyto_zoo_model}) always exists and is unstable. Zooplankton-free equilibrium solution $S_1=\displaystyle(0, K)$, which exists and stable when $\displaystyle C_Z<1$ with $\displaystyle C_P>1$ and unstable for $\displaystyle C_Z>1$. Also, the co-existence equilibrium solution $S^*=\displaystyle(Z^*, P^*)$ exists when $\displaystyle C_Z>1$. Moreover, for $\displaystyle C_Z>1$, there exists a critical value of $K^c$,
\begin{equation}
   \displaystyle K^c=\displaystyle\frac{h_p(\eta\alpha+d_z)}{(\eta\alpha-d_z)}
   \label{K^c}
\end{equation}
such that the co-existence equilibrium solution $S^*$ is\\
(i) stable when $K<K^c$,\\
(ii) center when $K=K^c$,\\
(iii) unstable associated with the appearance of bounded periodic solutions with initial amplitude and period of $\displaystyle\exp{\Big(\frac{r_p}{2 C_Z}-\frac{r_ph_p}{K(C_P-1)}\Big)}$ and\\
$\displaystyle 4\pi / \displaystyle \sqrt{4r_pd_z\big(1-\frac{1}{C_Z}\big)- r_p^2\Big(\frac{1}{C_Z}-\frac{2r_ph_p}{K(C_P-1)}\Big)^2}$, respectively  for $K>K^c$.
\end{lemma}
\noindent Note that $K^c$ represents the critical phytoplankton carrying capacity at which the coexistence equilibrium in the phytoplankton-zooplankton system~\eqref{phyto_zoo_model} becomes unstable.

%#############################################
%############################################
\section{Outbreak dynamics}
\label{outbreak_appendix}
%%%%%%%%%%%%%%%%%%%%%%%
\subsection{\textbf{Condition for initial outbreak growth}}\label{out_cond_cal}
The basic characteristics of an outbreak at an initial stage can be inferred from the sign of $dI/dt|_{t=0}$. The outbreak will initially grow, remain stationary, or decay depending on whether $dI/dt|_{t=0}$ is greater than, equal to, or less than zero, respectively. From the 2nd equation of model~(2.1), we have the necessary condition for the initial growth of an outbreak as
\begin{equation}
    \displaystyle \frac{\beta S_0B_0}{h_b + B_0} + \frac{\beta_z S_0Z_{B_0}}{h_z + Z_{B_0}} > (\gamma + \delta)I_0.
\label{eq0}
\end{equation}
Here $N_0=S_0+I_0$ is the initial human population size, where $S_0$ and $I_0$ represent the initial susceptible and infected population, respectively. From the 4th \& 5th equation of model~(2.1), the initial growth in $B, Z_B$ cells implies $dB/dt>0$, $dZ_B/dt>0$, which gives
\begin{equation*}
    \displaystyle  \xi I_0 -d_b B_0- c\sigma \frac{B_0Z_F^*}{h_m+B_0} >0 ~~\text{and} ~~ \displaystyle  \sigma \frac{B_0Z_F^*}{h_m+B_0} -d_zZ_{B_0} >0.
\end{equation*}
After simplification, we have
\begin{equation}
\displaystyle I_0> \frac{d_b B_0}{\xi}+c\sigma \frac{B_0Z_F^*}{\xi(h_m+B_0)} ~~\text{and}~~ Z_{B_0}<\sigma\frac{B_0Z_F^*}{d_z(h_m+B_0)}.
\label{eq1}
\end{equation}
Using Eq.~\eqref{eq1}, from Eq.~\eqref{eq0}, we have
\begin{equation*}
    \displaystyle \frac{\beta S_0B_0}{h_b + B_0} + \frac{\beta_z S_0}{h_z + Z_{B_0}}\sigma\frac{B_0Z_F^*}{d_z(h_m+B_0)}  > \frac{(\gamma + \delta)}{\xi}\Big(d_b B_0+c\sigma \frac{B_0Z_F^*}{h_m+B_0}\Big)
\end{equation*}
\begin{equation*}
\displaystyle {\rm or},~~ \frac{\xi S_0}{(\gamma+\delta)}\Bigg(\frac{\beta}{h_b + B_0} + \frac{\beta_z}{h_z + Z_{B_0}}\sigma\frac{Z_F^*}{d_z(h_m+B_0)}\Bigg)> \Big(d_b+c\sigma \frac{Z_F^*}{h_m+B_0}\Big).
\label{eq2}
\end{equation*}
Now, at $t=0$, from above, $S_0\approx N_0$ and $B_0, Z_{B_0}\approx0, Z_{F_0}\approx Z^*$ give
\begin{equation}
\begin{array}{lll}
     & \displaystyle \frac{\xi N_0}{(\gamma+\delta)}\Bigg[\frac{\beta}{h_b} + \frac{\beta_z}{h_z}\sigma\frac{Z^*}{d_zh_m}\Bigg]> \Big(d_b+c\sigma \frac{Z^*}{h_m}\Big), \\\\
    or, & \displaystyle\frac{\xi N_0}{(\gamma+\delta)\Big(d_b+c\sigma \frac{Z^*}{h_m}\Big)}\Bigg[\frac{\beta}{h_b} + \frac{\beta_z}{h_z}\sigma\frac{Z^*}{d_zh_m}\Bigg]> 1,\\\\
    or, &  \mathcal{R}_{0_{\rm BZ}}^{\rm out} = \displaystyle \mathcal{R}_{0_B}^{\rm out}+\mathcal{R}_{0_Z}^{\rm out}>1
    % \implies & \displaystyle \mathcal{R}_0^{\rm out}>1, ~~\text{(say)}.
\end{array}
\label{R0BZ_app}
\end{equation}

\noindent The outbreak initially remains stationary or in a decaying state and subsequently decays whenever $\displaystyle \mathcal{R}_{0_{\rm BZ}}^{\rm out}\leq 1$. Here, $\displaystyle\frac{\xi}{(\delta+\gamma)}$ represents the average amount of pathogen shed by an infected person throughout their infectious period. 
% The environmental persistence of free-living \emph{Vibrios}, considering both natural mortality and association with zooplankton, is given by $\displaystyle\frac{1}{(d_b +\frac{c\sigma Z^*}{h_m})}$. The term $\displaystyle \frac{\sigma Z^*}{h_m}/\Big(d_b +\frac{c\sigma Z^*}{h_m}\Big)$ represents the fraction of \emph{Vibrios} that survive in the free-living state before colonizing zooplankton.

%##############################################
% \subsection{Derivation of $\beta_z^{c_1}$ and $\beta_z^{c_2}$}
\subsection{\textbf{Derivation of $\beta_z$ on $\displaystyle \mathcal{R}_{0_{\rm BZ}}^{\rm out}=1$ and $\displaystyle \mathcal{R}_{0_{\rm BZ}}^{\rm out}=\displaystyle \mathcal{R}_{0}^{\rm out}$ line}}\label{app_out_sec2}

From Eq.~\eqref{R0BZ_app}, we have the critical zooplankton-mediated transmission rate, $\beta_z^{c_1}$, sufficient to initiate an outbreak ($\displaystyle \mathcal{R}_{0_{\rm BZ}}^{\rm out}>1$) for each specific $\beta$, as given by:
\begin{equation}
  \beta_z^{c_1} = \Bigg(\frac{\Big(d_b +\displaystyle \frac{c\sigma Z^*}{h_m}\Big)(\delta+\gamma)}{\xi N_0}- \frac{\beta}{h_b}\Bigg)\frac{h_m h_z d_z}{\sigma Z^*}
  \label{appendix-beta_zc1}
    \end{equation}
% It is noteworthy that $\beta_z^*$ decreases as $\beta$ increases, indicating that even minimal exposure to zooplankton can alter the initial outbreak dynamics when free-living bacterial transmission is more pronounced.
Depending on the bacteria-zooplankton ($B$-$Z$) association rate ($\sigma$) and the transmission rate via zooplankton ($\beta_z$), $\displaystyle \mathcal{R}_{0_{\rm BZ}}^{\rm out}$ can be lower or higher than $\displaystyle \mathcal{R}_0^{\rm out}(= \mathcal{R}_{0_{\rm BZ}}^{\rm out}|_{\sigma = 0}$).\\
Now, on the $\displaystyle \mathcal{R}_{0_{\rm BZ}}^{\rm out}=\displaystyle \mathcal{R}_{0}^{\rm out}$ line, we have
\begin{equation*}
    \begin{array}{llll}
       \displaystyle \mathcal{R}_{0_{\rm BZ}}^{\rm out} = \mathcal{R}_0^{\rm out} &  \\\\
       \text{or,}~~ \displaystyle\frac{\xi N_0}{(\gamma+\delta)\Big(d_b+c\sigma \frac{Z^*}{h_m}\Big)}\Bigg[\frac{\beta}{h_b} + \frac{\beta_z}{h_z}\sigma\frac{Z^*}{d_zh_m}\Bigg] =~ \displaystyle\frac{\xi N_0}{d_b(\gamma+\delta)}\frac{\beta}{h_b} & \\\\
       \text{or,}~~ \displaystyle \frac{1}{\Big(d_b+c\sigma \frac{Z^*}{h_m}\Big)} \frac{\beta_z}{h_z}\sigma\frac{Z^*}{d_z h_m} = \displaystyle\frac{\beta}{h_b} \Bigg[\frac{1}{d_b}-\frac{1}{\Big(d_b+c\sigma \frac{Z^*}{h_m}\Big)}\Bigg] & \\\\
       \text{or,}~~ \displaystyle \frac{1}{\Big(d_b+c\sigma \frac{Z^*}{h_m}\Big)} \Bigg(\frac{\beta_z \sigma Z^*}{h_z d_z h_m} - \frac{\beta c\sigma Z^*}{h_b d_b h_m}\Bigg) = 0 & \\\\
       \text{or,}~~ \displaystyle \beta_z = \displaystyle \frac{c h_z d_z}{h_b d_b}\beta = \beta_z^{c_2}. & 
    \end{array}
\end{equation*}
Now, $\displaystyle\frac{\partial \beta_z^{c_1}}{\partial \beta}<0$ and $\displaystyle\frac{\partial \beta_z^{c_2}}{\partial \beta}>0$ indicate that at an increased $\beta$, a lower $\beta_z^{c_1}$ is required for the initial outbreak growth, whereas a higher $\beta_z^{c_2}$ is necessary for $\displaystyle \mathcal{R}_{0_{\rm BZ}}^{\rm out}>\mathcal{R}_0^{\rm out}$.\\

\noindent On the line $\mathcal{R}_0^{\rm out}=1$ \Big(i.e., $\displaystyle \beta=\frac{(\gamma+\delta)h_bd_b}{\xi N_0}$\Big), we have
\begin{equation*}
\begin{array}{ll}
    \displaystyle \beta_z^{c_1} &=\displaystyle \frac{(\gamma+\delta)}{\xi N_0}\Bigg(\Big(d_b + \frac{c\sigma Z^*}{h_m}\Big)- \frac{\xi N_0}{(\gamma+\delta)}\frac{\beta}{h_b}\Bigg)\frac{h_m h_z d_z}{\sigma Z^*}  \\\\
   &= \displaystyle \frac{(\gamma+\delta)}{\xi N_0}\Bigg(\Big(d_b +\frac{c\sigma Z^*}{h_m}\Big)- d_b\Bigg)\frac{h_m h_z d_z}{\sigma Z^*} = \frac{(\gamma+\delta)ch_zd_z}{\xi N_0}.
\end{array}
    \end{equation*}
Also, $\beta_z^{c_2}=\displaystyle\frac{ch_zd_z}{h_bd_b}\frac{(\gamma+\delta)h_bd_b}{\xi N_0}= \frac{(\gamma+\delta)ch_zd_z}{\xi N_0}$.\\\\
Since $\beta_z^{c_2}$ satisfies $\displaystyle \mathcal{R}_{0_{\rm BZ}}^{\rm out}= \mathcal{R}_0^{\rm out}$, so we have $\beta_z^{c_1}=\beta_z^{c_2}$ at $\mathcal{R}_0^{\rm out}=\mathcal{R}_{0_{\rm BZ}}^{\rm out}=1$.

\subsection{\textbf{Effect of $\sigma$ on $\mathcal{R}_{0_{\rm BZ}}^{\rm out}$ within the ($\beta$-$\beta_z$) parameter space}}
Now, \begin{equation*}
    \begin{array}{lll}
        \displaystyle\frac{\partial \beta_z^{c_1}}{\partial \sigma}&=\displaystyle \frac{h_mh_zd_z}{Z^*}\Big(-\frac{(\gamma+\delta)}{\xi N_0}\frac{d_b}{\sigma^2} + \frac{\beta}{h_b\sigma^2}\Big)  \\\\
        &= \displaystyle \frac{h_mh_zd_z}{h_b\sigma^2Z^*}\Big(\beta-\frac{(\gamma+\delta)h_bd_b}{\xi N_0}\Big) \\\\
        &= \displaystyle \frac{h_mh_zd_z}{h_b\sigma^2Z^*}\frac{(\gamma+\delta)h_bd_b}{\xi N_0}\Big(\mathcal{R}_0^{\rm out}-1\Big) 
    \end{array}
\end{equation*}
Thus, a higher $\sigma$ decreases $\beta_z^{c_1}$ when $\mathcal{R}_0^{\rm out}<1$ and increases $\beta_z^{c_1}$ when $\mathcal{R}_0^{\rm out}>1$. In other words, $\sigma$ decreases the slope of the line $\displaystyle \mathcal{R}_{0_{\rm BZ}}^{\rm out}=1$. As a result, regions \boxed{6}, \boxed{3} expand at the expense of \boxed{5}, \boxed{2}, respectively, with higher $\sigma$ (see Fig.~2).\\

\noindent Now, \begin{equation*}
\begin{array}{ll}
 \displaystyle \frac{\partial\mathcal{R}_{0_{\rm BZ}}^{\rm out}}{\partial\sigma} &= \displaystyle\frac{\xi N_0}{(\gamma+\delta)} \frac{\frac{Z^*}{h_m}}{\Big(d_b+c\sigma \frac{Z^*}{h_m}\Big)^2} \Bigg[\frac{\beta_z d_b}{h_z d_z}-\frac{c\beta}{h_b}\Bigg]\\\\
 &=\displaystyle \frac{\xi N_0}{(\gamma+\delta)} \frac{\frac{Z^*}{h_m}}{\Big(d_b+c\sigma \frac{Z^*}{h_m}\Big)^2}\frac{d_b}{h_zd_z} \Big(\beta_z-\beta_z^{c_2}\Big) \\\\
 &\displaystyle >, ~= ~\text{or} < 0 ~~~\text{iff}~ \beta_z >, ~= ~\text{or} < \beta_z^{c_2}. 
\end{array}
\end{equation*}
The decrease (within \boxed{2}, \boxed{3}, \boxed{4}) or increase (within \boxed{5}, \boxed{6}, \boxed{1}) of $\mathcal{R}_{0_{\rm BZ}}^{\rm out}$ with higher $\sigma$, depending on whether $\beta_z<\beta_z^{c_2}$ or $\beta_z>\beta_z^{c_2}$, can be observed in Fig.~2(b).

%############################################
%#########################################
\subsection{\textbf{Sensitivity analysis}}
\label{sensitivity}
 \begin{figure}[h!]
\begin{center}
{\includegraphics[width=0.7\textwidth]{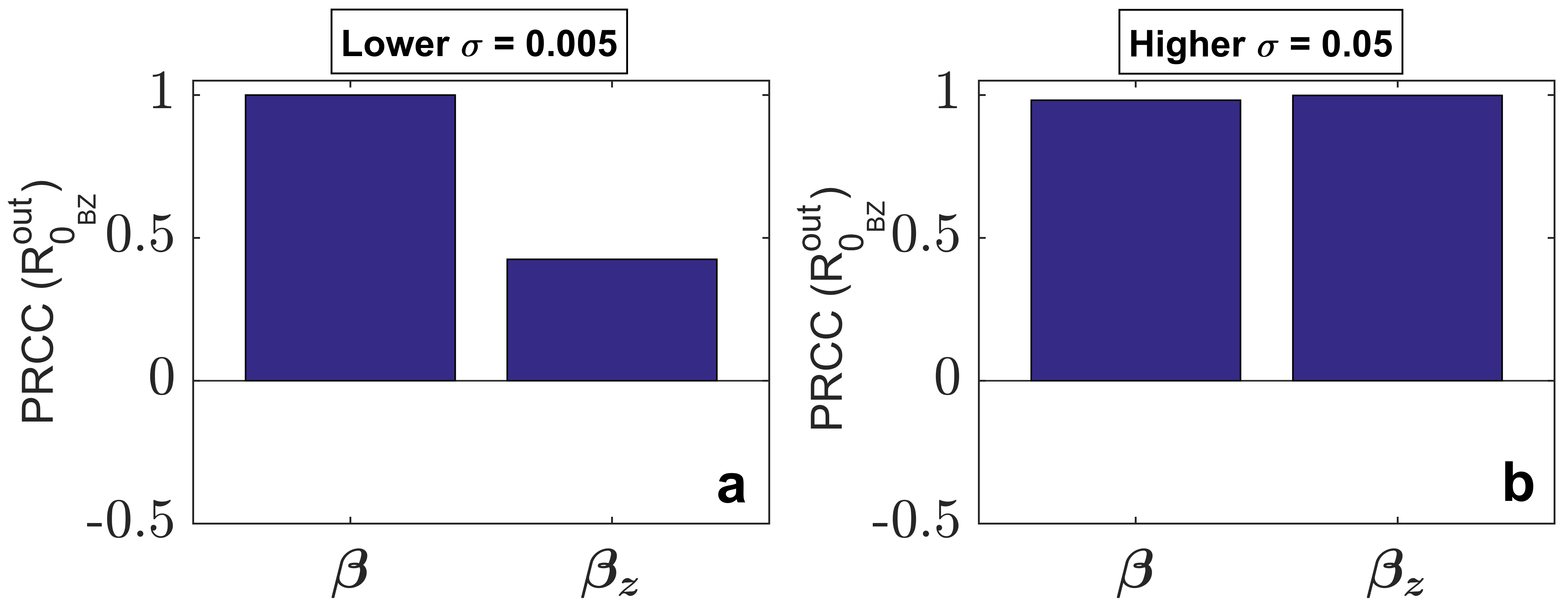}}
\end{center}
\caption{Sensitivity indices for $\mathcal{R}_{0_{\rm BZ}}^{\rm out}$ w.r.t $\beta$, $\beta_z$ under low $\sigma$ ($\sigma=0.005$)(a) and high $\sigma$ ($\sigma=0.05$)(b).}
\label{figS1-prcc_R0}
\end{figure}

We perform a global sensitivity analysis of $\mathcal{R}_{0_{\rm BZ}}^{\rm out}$ with respect to $\beta$ and $\beta_z$ under both low and high $\sigma$, using Latin hypercube sampling (LHS) coupled with partial rank correlation coefficients (PRCCs) \citep{saltelli2004sensitivity, marino2008methodology}. We draw 1,000 samples for each parameter using the LHS scheme, considering a $20\%$ range of variation around baseline values given in Table~1. Under high $\sigma$, $\beta_z$ shows greater sensitivity to $\mathcal{R}_{0_{\rm BZ}}^{\rm out}$ compared to lower $\sigma$, while as expected, $\beta$ remains consistently more sensitive in both cases (see Fig.~\ref{figS1-prcc_R0}).

% Here, we consider $\mathcal{R}_0^{\rm out}$ and the final size as response functions. Since the effect of bacteria-zooplankton association on disease severity is highly dependent on the association rate $\sigma$, we perform sensitivity analysis under both lower ($\sigma=0.005$) and higher ($\sigma=0.05$) values of $\sigma$. We draw 1,000 samples for each parameter using the LHS scheme, considering a $20\%$ range of variation around baseline values given in Table~\ref{table-1}. Note that, to avoid negative feedback of $\sigma$ on $\mathcal{R}_0^{\rm out}$, we keep $\beta_z$ greater than $\beta_z^c$ for each combination of $\beta$ and $\beta_z$.
% {\color{purple} The nonlinear and monotone relationship between the response functions and the input parameters across their variation range is displayed in Fig.~\ref{prcc_scatter}, which is a prerequisite for computing the PRCCs.} 
% The nonlinear and monotone relationship between the response functions and the input parameters across their range of variation is observed, which is a prerequisite for computing the PRCCs. It is worth noting that local sensitivity analysis techniques, which investigate the impact of varying a single input parameter, may not be appropriate for a comprehensive quantitative analysis.

% #####################################
\begin{figure}[h!]
    \centering
    \includegraphics[width=0.5\linewidth]{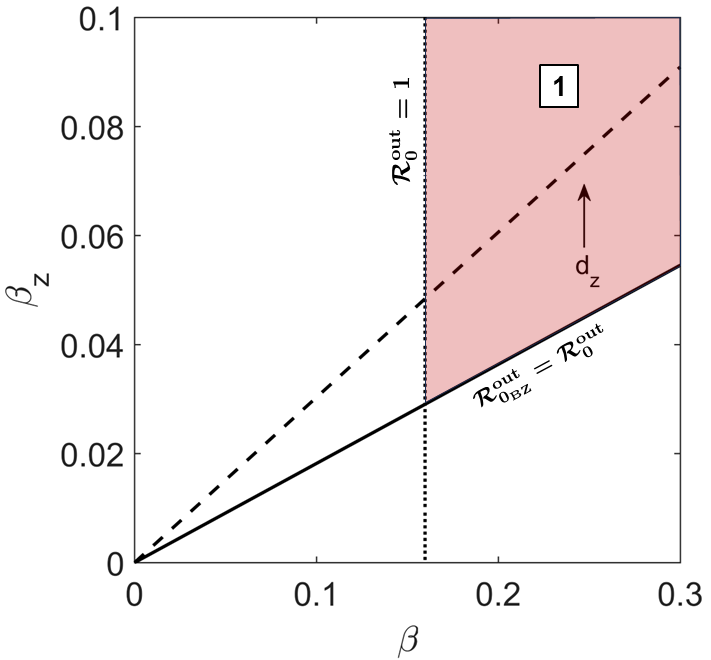}
    \caption{Zooplankton death rate ($d_z$) increases the slope of the $\mathcal{R}_{0_{\rm BZ}}^{\rm out}=\mathcal{R}_0^{\rm out}$ line and thereby decreasing region \boxed{1}. Here, for the solid line, $d_z=0.06$, and for the dashed line, $d_z=0.1$.}
    \label{fig7-app_beta_vs_betaz_dz}
\end{figure}

% #################################################
\subsection{\textbf{Relative contribution of $\mathcal{R}_{0_B}^{\rm out}$, $\mathcal{R}_{0_Z}^{\rm out}$ to $\mathcal{R}_{0_{\rm BZ}}^{\rm out}$}}\label{app_rel_cont}

For a fixed $\beta$, to keep $\mathcal{R}_{0_{\rm BZ}}^{\rm out}$ same as $\mathcal{R}_0^{\rm out}$ \Big(i.e., $\displaystyle \beta_z =\frac{c h_z d_z}{h_b d_b}\beta$\Big), there exists a unique bacteria-zooplankton association rate ($\sigma$) for each pair of relative contributions ($\mathcal{R}_{0_B}^{\rm out}$, $\mathcal{R}_{0_Z}^{\rm out}$).\\
% Let $\mathcal{R}_{0_Z}^{\rm out}/\mathcal{R}_{0_{\rm BZ}}^{\rm out}=x$. 
Now,
\begin{equation*}
    \begin{array}{ll}
      \displaystyle\mathcal{R}_{0_B}^{\rm out}&= \displaystyle\frac{\xi N_0}{(\gamma+\delta)\Big(d_b+c\sigma \frac{Z^*}{h_m}\Big)}\frac{\beta}{h_b} \\\\
  \text{~or,~}  \displaystyle \sigma&=\displaystyle\frac{h_m}{c Z^*}\Bigg(\frac{\xi N_0}{(\gamma+\delta)h_b}\frac{\beta}{h_b}\frac{1}{\mathcal{R}_{0_B}^{\rm out}}-d_b \Bigg)\\\\
  &=\displaystyle\frac{d_bh_m}{c Z^*}\Bigg(\frac{\mathcal{R}_{0_{\rm BZ}}^{\rm out}}{\mathcal{R}_{0_B}^{\rm out}}-1 \Bigg) = \displaystyle\frac{d_bh_m}{c Z^*}\frac{\mathcal{R}_{0_Z}^{\rm out}}{\mathcal{R}_{0_B}^{\rm out}}=\displaystyle\frac{d_bh_m}{c Z^*}\Bigg(\frac{1}{\mathcal{R}_{0_{\rm BZ}}^{\rm out}/\mathcal{R}_{0_{Z}}^{\rm out}-1}\Bigg)\\\\
  \text{~or,~}  \displaystyle \frac{\mathcal{R}_{0_{Z}}^{\rm out}}{\mathcal{R}_{0_{\rm BZ}}^{\rm out}}&= \displaystyle \frac{c\sigma Z^*}{d_bh_m+c\sigma Z^*}.
    \end{array}
\end{equation*}
% Using above $\sigma$, from the expression\\
% $\displaystyle \Bigg(\frac{\mathcal{R}_{0_{\rm BZ}}^{\rm out}}{\mathcal{R}_{0_B}^{\rm out}}-1\Bigg)=\mathcal{R}_{0_Z}^{\rm out} /\mathcal{R}_{0_B}^{\rm out} = (\sigma Z^* h_b \beta_z)/(d_Z h_z h_m \beta)$,\\
% we can have $\displaystyle \beta_z =\frac{c h_z d_z}{h_b d_b}\beta= \beta_z^{c_2}$.

% #############################################
\subsection{\textbf{Epidemic growth rate (EGR)}}
\label{EGR}
The initial exponential growth rate (EGR) of an epidemic assesses the early outbreak dynamics and is crucial in inferring the basic reproduction number. The EGR is determined by the dominant eigenvalue of the Jacobian at the disease-free equilibrium point $E^0=(N_0, 0, 0, 0, 0, Z^*, P^*)$ \citep{tien2010multiple}. For system~(2.1), the initial EGR, denoted by $\lambda$, can be obtained from the following equation:
\begin{equation*}
     \displaystyle x^3+c_2 x^2+c_1 x+c_0=0
     \label{cubic_lambda}
\end{equation*}
Where the coefficients are given by
\begin{eqnarray*}
\begin{array}{lll}
      \displaystyle c_2&=&\displaystyle d_b+\frac{c \sigma Z^*}{h_m}+\gamma+\delta+d_z>0,\\[1.5ex]
      \displaystyle c_1&=&\displaystyle\Big(\frac{c \sigma Z^*}{h_m}+d_b+d_z\Big)(\gamma+\delta)+\Big(\frac{c \sigma Z^*}{h_m}+d_b\Big) d_z-\frac{N_0\xi \beta}{h_b},\\[1.5ex]
      \displaystyle c_0&=&\displaystyle(\gamma+\delta)\Big(\frac{c \sigma Z^*}{h_m}+d_b\Big) d_z- \frac{N_0 \xi \beta d_z}{h_b}- \frac{N_0 \xi \beta_z}{h_z}\frac{\sigma Z^*}{h_m},\\[1.5ex]
      \displaystyle&=&\displaystyle(\gamma+\delta)\Big(\frac{c \sigma Z^*}{h_m}+d_b\Big) d_z \Big(1-\mathcal{R}_{0_{\rm BZ}}^{\rm out} \big).
\end{array}
\end{eqnarray*}
\begin{figure}[h!]
\begin{center}
{\includegraphics[width=0.6\textwidth]{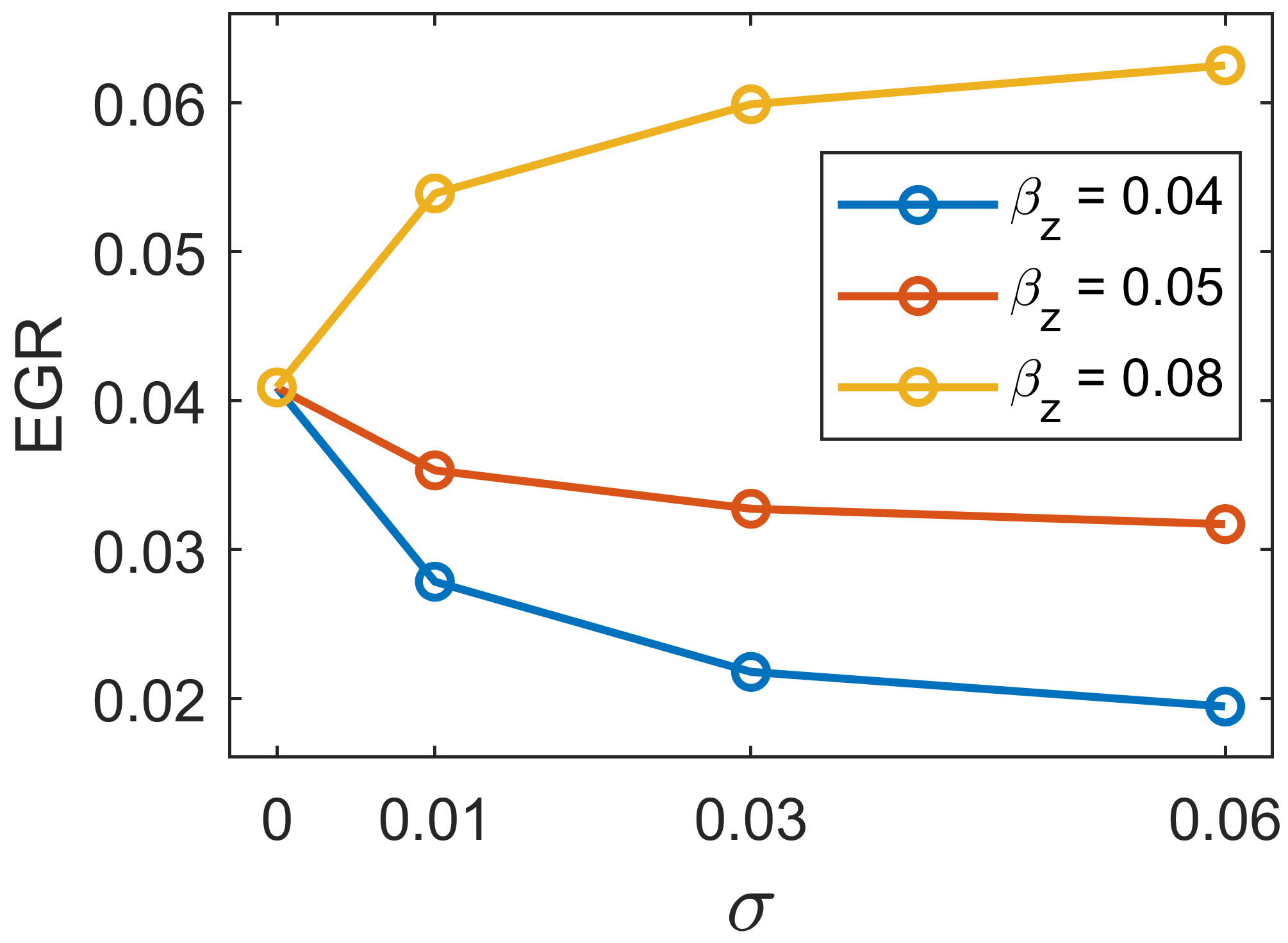}}
\end{center}
\caption{Initial epidemic growth rate (EGR) with respect to the \textit{Vibrio}-zooplankton association rate ($\sigma$) for different zooplankton-mediated transmission rate ($\beta_z$) within region \boxed{1} of Fig.~2(a).}
\label{FigS2-sigma_vs_EGR}
\end{figure}
For $\mathcal{R}_{0_{\rm BZ}}^{\rm out}<1$, we have $c_0>0$ and can also show that $c_1>0$. Using the Descartes' rule of signs, we have $\lambda < 0$, which indicates that the outbreak will not initiate with a few initial infections. However, for $\mathcal{R}_{0_{\rm BZ}}^{\rm out}>1$, we get $c_0<0$. In this case, regardless of the sign of $c_1$, we have $\lambda > 0$, implying that the disease has the potential to invade the population.\\

The impact of increasing $\sigma$ on the EGR depends on the value of $\beta_z$ within region \boxed{1} of Fig.~2(a) (see Fig.~\ref{FigS2-sigma_vs_EGR}). 
% {\color{red}Additionally, altering the relative contributions of transmission routes to fixed a $\mathcal{R}_{0_{\rm BZ}}^{\rm out}$ impacts the EGR, thereby affecting the outbreak progression (Fig.~3(d)).}
% {\color{blue}Mention from BMB-2010 that in the presence of environmental pathogens initially, the final size relation may not be satisfied for multiple routes. Thus, it is possible that there may be different final sizes depending on the relative contribution of routes for the same R0.\\
% Read "A Note on the Derivation of Epidemic Final Sizes" and "The Final Size of an Epidemic and Its Relation to the Basic Reproduction Number".}

% ####################################

\subsection{\textbf{Post-peak maintenance of low-level infections}}
Here, we investigate the impact of varying transmission rates via zooplankton ($\beta_z$) on outbreak trajectories, under a fixed bacteria-zooplankton association rate ($\sigma$) within \boxed{1} in Fig.~2. Zooplankton serves as a reservoir for \emph{Vibrios}, helping to sustain lower-level infections over an extended post-peak phase before finally dying out (see Fig.~\ref{figS3-time_series}).
\begin{figure}[h!]
\begin{center}
\includegraphics[width=0.99\textwidth]{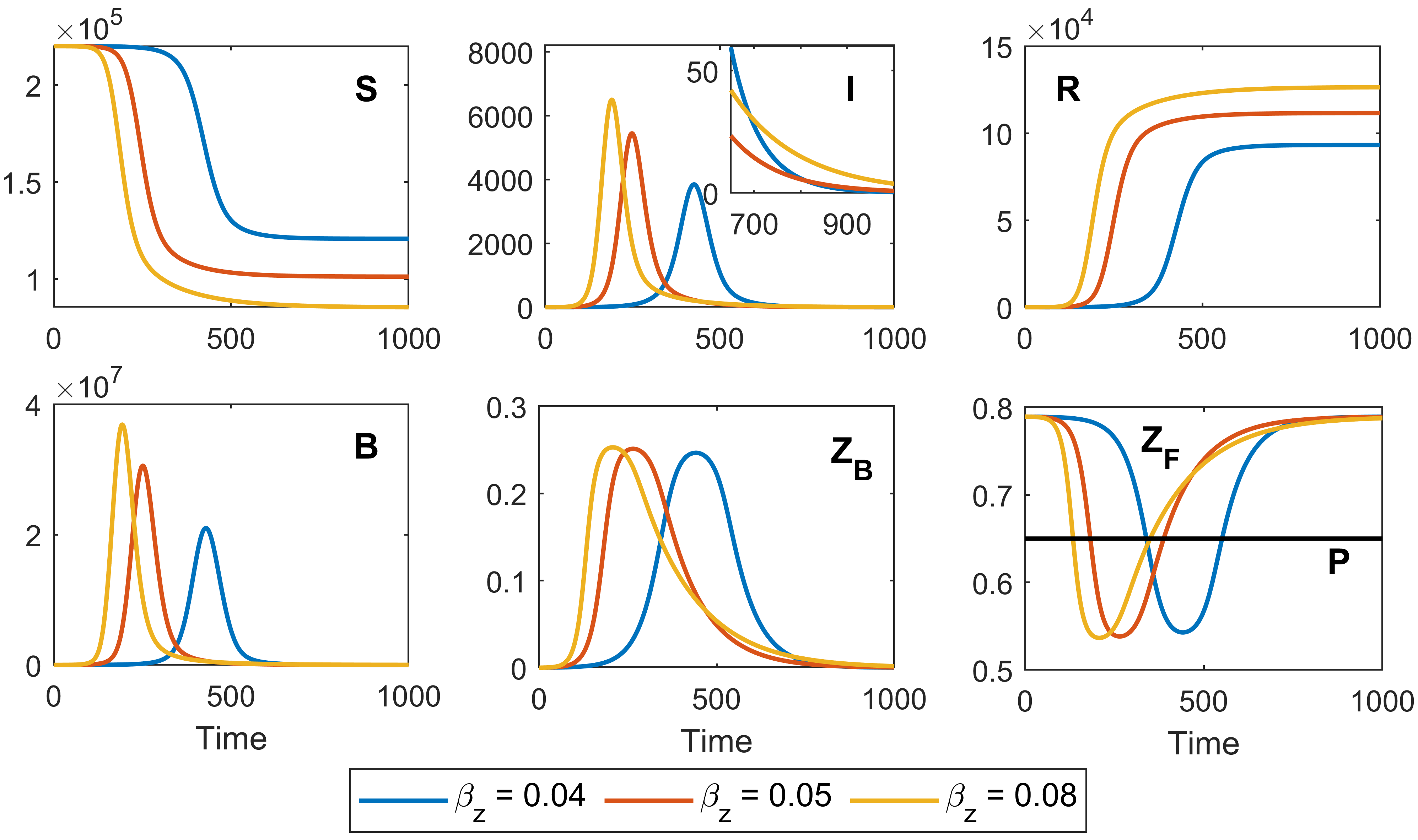}
\end{center}
\caption{Increased zooplankton-mediated transmission ($\beta_z$) can have potentially large negative impacts on human health under a fixed bacteria-zooplankton association rate ($\sigma=0.03$). It shortens the time to peak infection, increases the peak size, and elevates lower-level maintenance of infections during the post-peak period. %The epidemic dies out at around 500 days when $\sigma, \beta_z=0$ (blue line). However, for non-zero $\sigma, \beta_z$, the disease sustained at a low level for an extended duration (with earlier and larger peak), exceeding 700 days (red, yellow lines). 
}
\label{figS3-time_series}
\end{figure}

\begin{figure}[h!]
\begin{center}
{\includegraphics[width=0.85\textwidth]{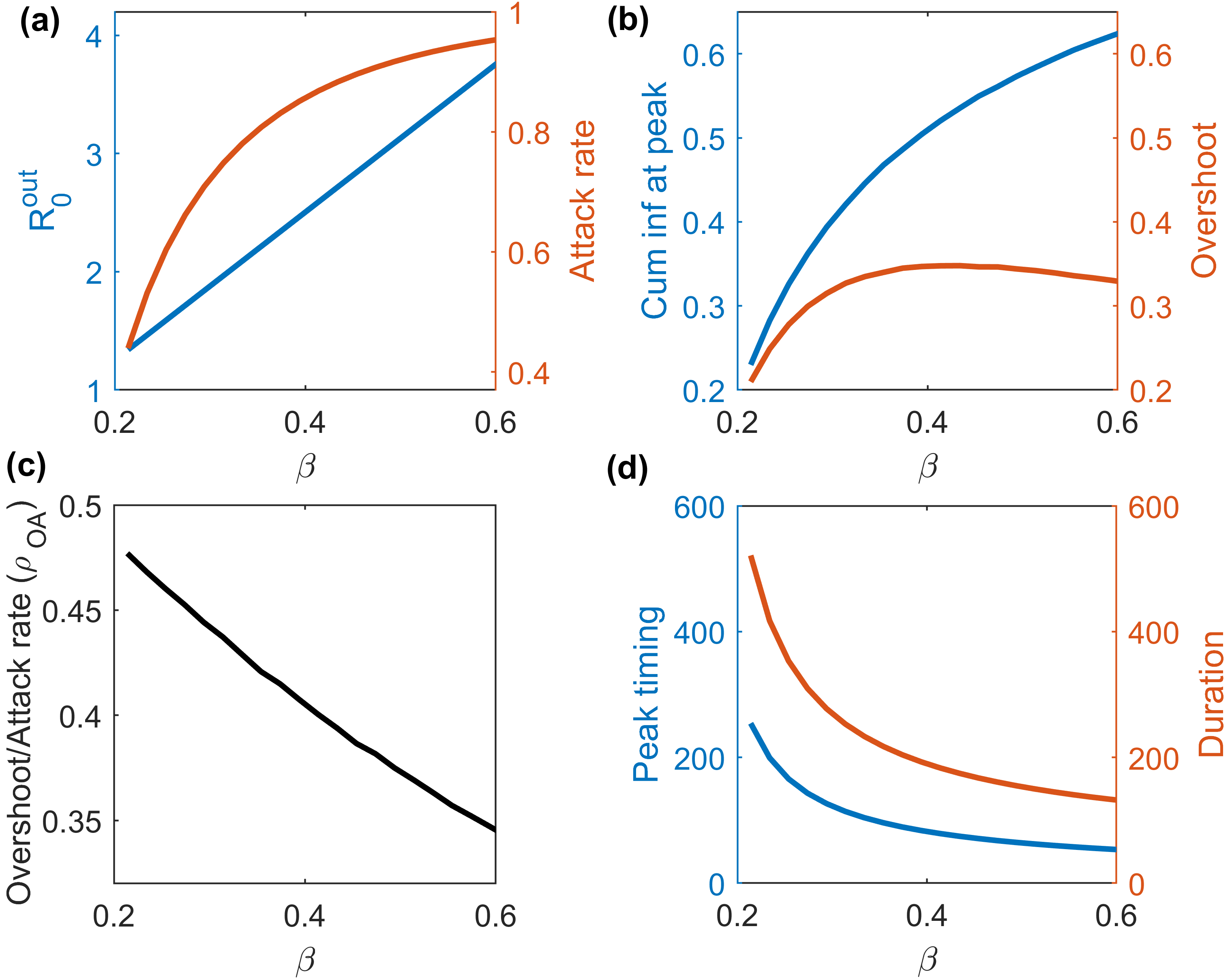}}
\end{center}
\caption{Effect of transmission rate via the free-living bacterial route ($\beta$) for the classical SIRB model on (a) $\mathcal{R}_0^{\rm out}$ and attack rate, (b) cumulative infections at peak (as a proportion of the total population) and epidemic overshoot, (c) the ratio of overshoot to attack rate ($\rho_{\rm OA}$), and (d) peak timing and epidemic duration in the absence of bacteria-zooplankton association ($\sigma=0$). Intuitively, both $\mathcal{R}_0^{\rm out}$ and attack rate increase with $\beta$. While $\beta$ increases cumulative infections at peak, there exists an upper bound on the overshoot. The ratio $\rho_{\rm OA}$ decreases monotonically with increasing $\beta$ due to the availability of fewer susceptible individuals for the overshoot (post-peak) phase. Both peak timing and epidemic duration consistently decrease as $\beta$ increases.}
\label{fig8-app_beta_vs_overshoot}
\end{figure}

% ####################################
% #################################
\section{Long-term dynamics}
% \subsection{System equilibria and their feasibility}
 %  The system~(2.1) consists of the following biologically feasible equilibrium points.
 %  \begin{itemize}
 %      \item Trivial equilibrium point, $\displaystyle E_0^0=(\frac{\Lambda}{\mu}, 0, 0, 0, 0, 0, 0)$
 %      \item Bacteria-zooplankton-free equilibrium point, $\displaystyle E_1^0=(\frac{\Lambda}{\mu}, 0, 0, 0, 0, 0, K)$
 %      \item Disease-free equilibrium point, $\displaystyle E^0=(\frac{\Lambda}{\mu}, 0, 0, 0, 0, Z^*, P^*)$
 %      \item Plankton-free endemic equilibrium point, $\displaystyle E_0^*=(S_0^*, I_0^*, R_0^*, B_0^*, 0, 0, 0)$
 %      \item Zooplankton-free endemic equilibrium point, $\displaystyle E_1^*=(S_1^*, I_1^*, R_1^*, B_1^*, 0, 0, K)$
 %      \item Interior equilibrium point, $\displaystyle E^*=(S^*, I^*, R^*, B^*, Z_B^*, Z^*, P^*)$
 % \end{itemize}

\subsection{\textbf{The basic reproduction number ($\mathcal{R}_{0_{\rm BZ}}^l$)}}
\label{R0_calculation}
 Using the next-generation matrix approach \citep{diekmann2010construction}, we obtain the basic reproduction number ($\mathcal{R}^l_{0_{\rm BZ}}$) in long-term scenario ($\Lambda,\mu,\omega\neq 0$), defined as the average number of secondary infections appearing from an average primary case within an entirely susceptible population.

\noindent When $C_Z>1$, the disease-free equilibrium (DFE) for the model~(2.1) is given by $\displaystyle E^0=(\frac{\Lambda}{\mu}, 0, 0, 0, 0, Z^*, P^*)$. 
% The basic reproduction number in long-term scenario can be calculated using the next-generation matrix approach \citep{diekmann2010construction}. 
If $X=(I,B,Z_B)^T$, which are the infected compartments, then we can write $\displaystyle \frac{dX}{dt} = (\mathcal{F}+\mathcal{V})X$ from Eq.~(2.1), where,

{\small
\begin{equation*}
 \mathcal{F}=\displaystyle \left(\begin{array}{ccc} 0 & \displaystyle \frac{\beta h_b S}{(h_b+B)^2} & \displaystyle \frac{\beta_z h_z S}{(h_z+Z_B)^2}\\0 & 0 & 0 \\0 & 0 & 0 \end{array}\right),
 \end{equation*}
 \begin{equation*}
\mathcal{V}= \left(\begin{array}{ccc}  -(\gamma +\mu+\delta) & 0 & 0 \\\xi &\displaystyle -d_b- c\sigma \frac{h_m Z_F}{(h_m+B)^2} & 0\\ 0 &\displaystyle \sigma \frac{h_m Z_F}{(h_m+B)^2}  &\displaystyle -d_z\end{array}\right).
\end{equation*}}

\noindent This gives,
{\small
\begin{equation*}
F= \mathcal{F}|_{E^0}=\displaystyle \left(\begin{array}{ccc} 0 & \displaystyle \frac{\beta \Lambda}{h_b\mu} & \displaystyle \frac{\beta_z \Lambda}{h_z\mu}\\0 & 0 & 0 \\0 & 0 & 0 \end{array}\right), 
~V=-\mathcal{V}|_{E^0}= \left(\begin{array}{ccc}  \gamma +\mu+\delta & 0 & 0 \\-\xi &\displaystyle d_b+ \frac{ c\sigma Z^*}{h_m} & 0\\ 0 &\displaystyle -\frac{\sigma Z^*}{h_m}  &\displaystyle d_z\end{array}\right).
\end{equation*}}

\noindent Now,
\begin{eqnarray*}
    \begin{array}{cc}
    FV^{-1}& \\
        =&\displaystyle \left(\begin{array}{ccc} 0 & \displaystyle \frac{\beta \Lambda}{h_b\mu} & \displaystyle \frac{\beta_z \Lambda}{h_z\mu}\\0 & 0 & 0 \\0 & 0 & 0 \end{array}\right)
\left(\begin{array}{ccc}  \displaystyle \frac{1}{\gamma +\mu+\delta} & 0 & 0 \\[2.5ex] \displaystyle  \frac{h_m \xi}{(d_b h_m + c\sigma Z^*)(\delta + \gamma + \mu)} &\displaystyle  \frac{h_m}{d_b h_m + \sigma Z^* }& 0\\[2.5ex] \displaystyle \frac{\xi \sigma Z^*}{d_z(d_b h_m + c\sigma Z^*)(\delta + \gamma + \mu)}& \displaystyle \frac{\sigma Z^*}{d_z (d_b h_m + \sigma Z^*)} &\displaystyle \frac{1}{d_z}\end{array}\right)
    \end{array}
\end{eqnarray*}

\noindent We have the basic reproduction number for the long-term scenario as
\begin{eqnarray}
    \begin{array}{ll}
\displaystyle\mathcal{R}^l_{0_{BZ}}& = \rho(FV^{-1})\\
    & = \displaystyle \frac{\xi \Lambda}{\mu (d_b +\displaystyle \frac{c\sigma Z^*}{h_m})(\delta+\gamma+\mu)}\Bigg(\frac{\beta}{h_b} + \frac{\sigma Z^*}{d_z h_m}\frac{\beta_z}{h_z}\Bigg)\\\\
    % & = \Bigg(\displaystyle \frac{\xi \Lambda}{\mu (d_b +\displaystyle \frac{\sigma Z^*}{h_m})(\delta+\gamma+\mu)} \frac{\beta}{h_b}\Bigg) + \Bigg(\displaystyle \frac{\xi \Lambda}{\mu (d_b +\displaystyle \frac{\sigma Z^*}{h_m})(\delta+\gamma+\mu)} \frac{\beta_z \sigma Z^*}{d_z h_z h_m}\Bigg)\\
    & = \displaystyle\mathcal{R}^l_{0_B}+\mathcal{R}^l_{0_Z}.
\end{array}
% \label{R_0}
\end{eqnarray}
Here, the components $\mathcal{R}^l_{0_B}$ and $\mathcal{R}^l_{0_Z}$ are associated with transmission via the $B$ and $Z_B$ route, respectively. 

% {\color{blue} Follow "Castillo-Chavez and Song" and Martcheva book for the forward bifurcation proof at $R_0=1$. Also, follow Medda-2024MCS. Draw figures to show the GAS of disease-free equilibrium and the stability of the interior point, as shown in Fig-3,4 in Kolaye2019.}

\subsection{\textbf{Impact of bacteria-zooplankton association on cholera endemicity}}\label{app_long_term_region}
% Environmental reservoirs of \textit{V. cholerae} can critically impact the cholera endemicity in the long-term scenario ($\Lambda,\mu,\omega\neq 0$). 
% The limit population size $\displaystyle N(t)$ approaches to $\frac{\Lambda}{\mu}$.

In Fig.~\ref{figS4-long_term_beta_vs_betaz}, we present different dynamical behaviors in the long-term endemic scenario ($\Lambda,\mu,\omega\neq 0$), similar to the outbreak scenario in Fig.~2. Regions \boxed{1}-\boxed{6}, separated by the three lines $\mathcal{R}_0^l=1$ (dashed vertical line), $\mathcal{R}_{0_{BZ}}^l=1$ (red line) and $\mathcal{R}_{0_{BZ}}^l=\mathcal{R}_0^l$ (black line), characterize the endemic persistence of cholera and the increased endemic severity for specific values of $\beta$ and $\sigma$ (Fig.~\ref{figS4-long_term_beta_vs_betaz}). Demographic factors can impact the size of the regions \boxed{1}-\boxed{6}. For instance, in (Fig.~\ref{figS4-long_term_beta_vs_betaz}(b)), it can be observed that with higher recruitment rates of susceptibles, even lower transmission rates can sustain cholera endemicity. This occurs because a higher constant influx of susceptible individuals, even at low transmission rates, fuels new infections, thereby maintaining endemicity.

\begin{figure}
    \centering
         \includegraphics[width=0.99\textwidth]{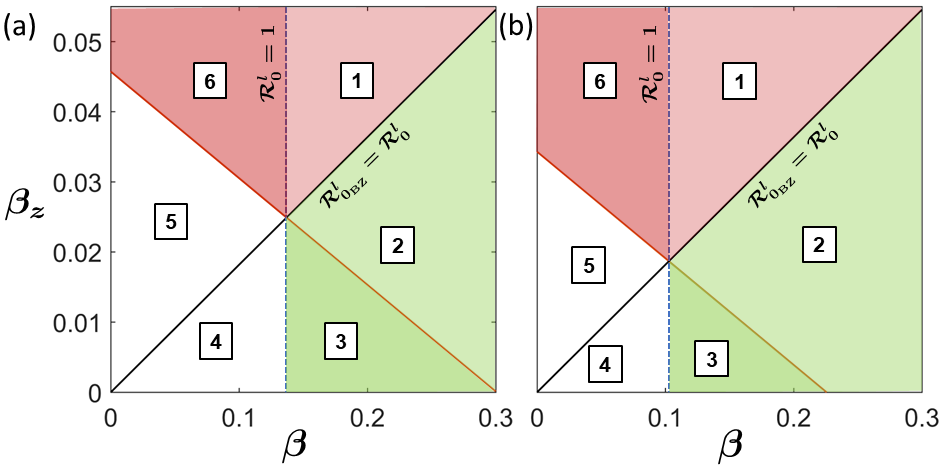}
    \caption{Lower values of $\beta$ and $\beta_z$ are sufficient to sustain cholera endemicity in a population with a higher recruitment rate ($\Lambda$). (a) $\Lambda=10.046$ and (b) $\Lambda=13.394$.}
    \label{figS4-long_term_beta_vs_betaz}
\end{figure}

%#######################################
\begin{figure}[h!]
\begin{center}
{\includegraphics[width=0.99\textwidth]{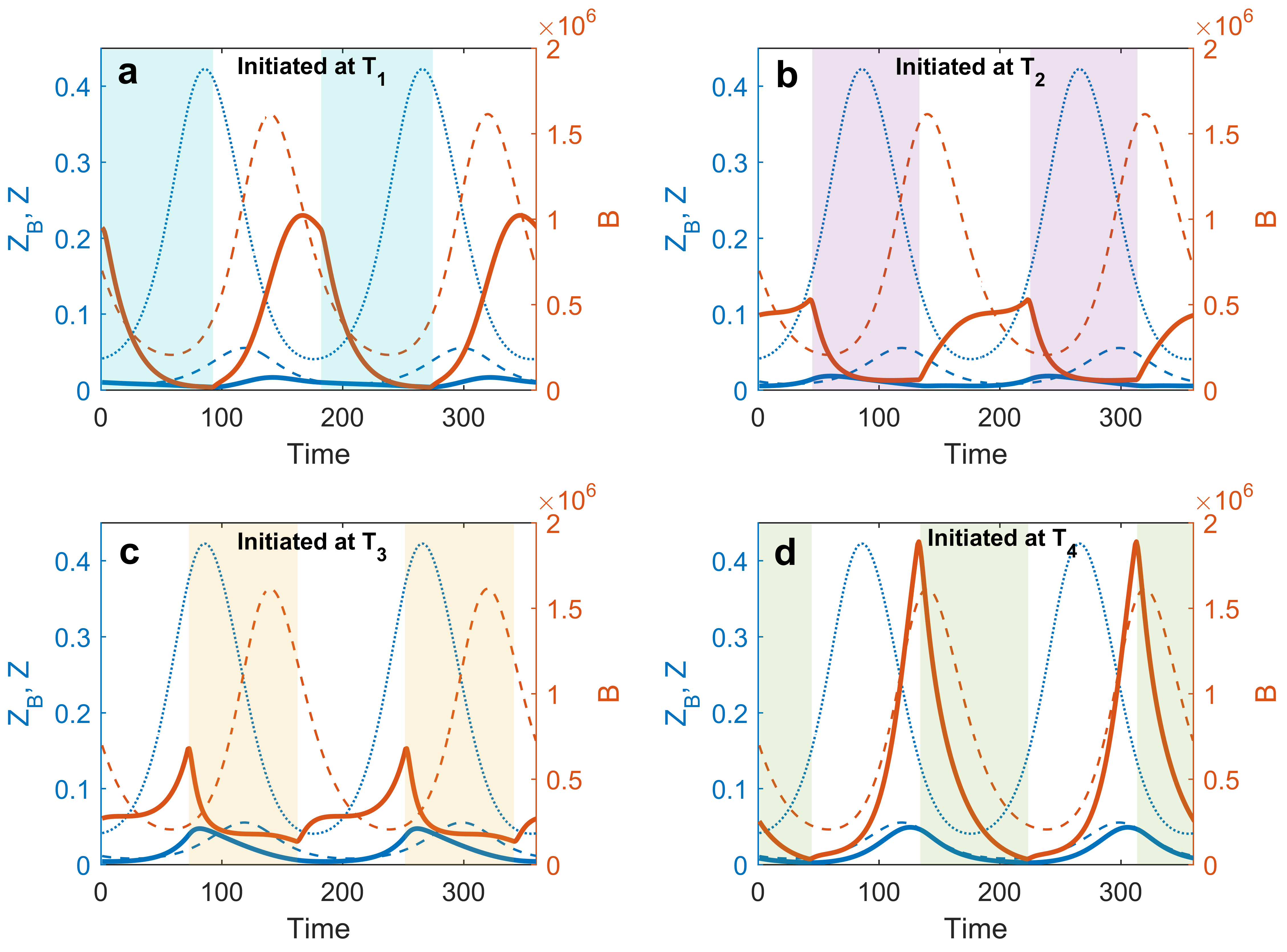}}
\end{center}
\caption{
% Abundance of $Z=(Z_F+Z_B)$ (dotted blue line), $Z_B$ (solid and dashed blue), $B$ (solid and dashed red) is shown for both unfiltered (dashed) and filtered (solid) scenarios when filtration is initiated at $T_1$-$T_4$. The formation of $Z_B$ depends on the density of $B$ cells in the water column during the period of $Z$ abundance.
The formation of $Z_B$ (solid and dashed blue) depends on the density of $B$ cells (solid and dashed red) in the water column during periods of $Z$ abundance (dotted blue line) in both unfiltered (dashed) and filtered (solid) scenarios when filtration is initiated at $T_1$-$T_4$. Filtration reduces cholera infections in two ways: first, by restricting $Z_B$ formation through the reduction of $B$ cell concentration during the periods $Z$ abundance; and second, if $Z_B$ formation cannot be avoided, by preventing its ingestion.
}
\label{fig9-app_filtration}
\end{figure}

\clearpage
\newpage

\section*{Author's contribution}
\textbf{B. M.:} Conceptualization, Methodology, Formal analysis, Software, Writing - original draft, Writing - review $\&$ editing. \textbf{S. B.:} Conceptualization, Methodology, Supervision, Writing - review $\&$ editing. \textbf{A. S.:} Conceptualization, Methodology, Supervision, Writing - review $\&$ editing. \textbf{J. P.:} Conceptualization, Writing - review $\&$ editing. \textbf{J. C.:} Supervision, Conceptualization, Writing - review $\&$ editing.

% \competing{The authors declare that they have no known competing financial interests or personal relationships that could have appeared to influence the work reported in this paper.}

\section*{Acknowledgements}
B. M. is supported by Junior Research Fellowship from University Grants Commission (UGC), India. A.S. would like to acknowledge Senior Research Fellowship from CSIR, India for funding him during the initial part of this work. S.B. acknowledges the funding via the Foundations and Applications of Emergence (FAEME) programme at the Dutch Institute for Emergent Phenomena (DIEP), University of Amsterdam and funding via the Marie Skłodowska–Curie grant agreement 101025056 for the project `SpatialSAVE'.

\bibliography{zoo_chol}

\begin{thebibliography}{81}
\expandafter\ifx\csname natexlab\endcsname\relax\def\natexlab#1{#1}\fi
\providecommand{\url}[1]{\texttt{#1}}
\providecommand{\href}[2]{#2}
\providecommand{\path}[1]{#1}
\providecommand{\DOIprefix}{doi:}
\providecommand{\ArXivprefix}{arXiv:}
\providecommand{\URLprefix}{URL: }
\providecommand{\Pubmedprefix}{pmid:}
\providecommand{\doi}[1]{\href{http://dx.doi.org/#1}{\path{#1}}}
\providecommand{\Pubmed}[1]{\href{pmid:#1}{\path{#1}}}
\providecommand{\bibinfo}[2]{#2}
\ifx\xfnm\relax \def\xfnm[#1]{\unskip,\space#1}\fi
%Type = Article
\bibitem[{Almagro-Moreno and Taylor(2013)}]{almagro2013cholera}
\bibinfo{author}{Almagro-Moreno, S.}, \bibinfo{author}{Taylor, R.K.},
  \bibinfo{year}{2013}.
\newblock \bibinfo{title}{Cholera: environmental reservoirs and impact on
  disease transmission}.
\newblock \bibinfo{journal}{Microbiology spectrum} \bibinfo{volume}{1}.
%Type = Article
\bibitem[{Amisu et~al.(2024)Amisu, Okesanya, Adigun, Manirambona, Ukoaka,
  Lawal, Idris, Olaleke, Okon, Ogaya et~al.}]{amisu2024cholera}
\bibinfo{author}{Amisu, B.O.}, \bibinfo{author}{Okesanya, O.J.},
  \bibinfo{author}{Adigun, O.A.}, \bibinfo{author}{Manirambona, E.},
  \bibinfo{author}{Ukoaka, B.M.}, \bibinfo{author}{Lawal, O.A.},
  \bibinfo{author}{Idris, N.B.}, \bibinfo{author}{Olaleke, N.O.},
  \bibinfo{author}{Okon, I.I.}, \bibinfo{author}{Ogaya, J.B.}, et~al.,
  \bibinfo{year}{2024}.
\newblock \bibinfo{title}{Cholera resurgence in africa: assessing progress,
  challenges, and public health response towards the 2030 global elimination
  target}.
\newblock \bibinfo{journal}{Le Infezioni in Medicina} \bibinfo{volume}{32},
  \bibinfo{pages}{148}.
%Type = Article
\bibitem[{Andrews and Basu(2011)}]{andrews2011transmission}
\bibinfo{author}{Andrews, J.}, \bibinfo{author}{Basu, S.},
  \bibinfo{year}{2011}.
\newblock \bibinfo{title}{Transmission dynamics and control of cholera in
  haiti: an epidemic model}.
\newblock \bibinfo{journal}{The Lancet} \bibinfo{volume}{377},
  \bibinfo{pages}{1248--1255}.
%Type = Article
\bibitem[{Baracchini et~al.(2017)Baracchini, King, Bouma, Rod{\'o}, Bertuzzo
  and Pascual}]{baracchini2017seasonality}
\bibinfo{author}{Baracchini, T.}, \bibinfo{author}{King, A.A.},
  \bibinfo{author}{Bouma, M.J.}, \bibinfo{author}{Rod{\'o}, X.},
  \bibinfo{author}{Bertuzzo, E.}, \bibinfo{author}{Pascual, M.},
  \bibinfo{year}{2017}.
\newblock \bibinfo{title}{Seasonality in cholera dynamics: A rainfall-driven
  model explains the wide range of patterns in endemic areas}.
\newblock \bibinfo{journal}{Advances in water resources} \bibinfo{volume}{108},
  \bibinfo{pages}{357--366}.
%Type = Article
\bibitem[{Chowdhury et~al.(1997)Chowdhury, Huq, Xu, Madeira and
  Colwell}]{chowdhury1997effect}
\bibinfo{author}{Chowdhury, M.}, \bibinfo{author}{Huq, A.},
  \bibinfo{author}{Xu, B.}, \bibinfo{author}{Madeira, F.},
  \bibinfo{author}{Colwell, R.R.}, \bibinfo{year}{1997}.
\newblock \bibinfo{title}{Effect of alum on free-living and copepod-associated
  vibrio cholerae o1 and o139}.
\newblock \bibinfo{journal}{Applied and environmental microbiology}
  \bibinfo{volume}{63}, \bibinfo{pages}{3323--3326}.
%Type = Article
\bibitem[{{Code\c{c}o, C.T.}(2001)}]{Codeco01}
\bibinfo{author}{{Code\c{c}o, C.T.}}, \bibinfo{year}{2001}.
\newblock \bibinfo{title}{{Endemic and epidemic dynamics of cholera: The role
  of the aquatic reservoir}}.
\newblock \bibinfo{journal}{BMC Infect. Dis.} \bibinfo{volume}{1:1}.
%Type = Article
\bibitem[{Colwell et~al.(1996)Colwell, Brayton, Herrington, Tall, Huq and
  Levine}]{colwell1996viable}
\bibinfo{author}{Colwell, R.}, \bibinfo{author}{Brayton, P.},
  \bibinfo{author}{Herrington, D.}, \bibinfo{author}{Tall, B.},
  \bibinfo{author}{Huq, A.}, \bibinfo{author}{Levine, M.},
  \bibinfo{year}{1996}.
\newblock \bibinfo{title}{Viable but non-culturable vibrio cholerae o1 revert
  to a cultivable state in the human intestine}.
\newblock \bibinfo{journal}{World Journal of Microbiology and biotechnology}
  \bibinfo{volume}{12}, \bibinfo{pages}{28--31}.
%Type = Article
\bibitem[{Colwell(1996)}]{colwell1996global}
\bibinfo{author}{Colwell, R.R.}, \bibinfo{year}{1996}.
\newblock \bibinfo{title}{Global climate and infectious disease: the cholera
  paradigm}.
\newblock \bibinfo{journal}{Science} \bibinfo{volume}{274},
  \bibinfo{pages}{2025--2031}.
%Type = Article
\bibitem[{Colwell and Huq(1994)}]{colwell1994environmental}
\bibinfo{author}{Colwell, R.R.}, \bibinfo{author}{Huq, A.},
  \bibinfo{year}{1994}.
\newblock \bibinfo{title}{Environmental reservoir of vibrio cholerae the
  causative agent of cholera a}.
\newblock \bibinfo{journal}{Annals of the New York Academy of Sciences}
  \bibinfo{volume}{740}, \bibinfo{pages}{44--54}.
%Type = Article
\bibitem[{Colwell et~al.(2003)Colwell, Huq, Islam, Aziz, Yunus, Khan, Mahmud,
  Sack, Nair, Chakraborty et~al.}]{colwell2003reduction}
\bibinfo{author}{Colwell, R.R.}, \bibinfo{author}{Huq, A.},
  \bibinfo{author}{Islam, M.S.}, \bibinfo{author}{Aziz, K.},
  \bibinfo{author}{Yunus, M.}, \bibinfo{author}{Khan, N.H.},
  \bibinfo{author}{Mahmud, A.}, \bibinfo{author}{Sack, R.B.},
  \bibinfo{author}{Nair, G.B.}, \bibinfo{author}{Chakraborty, J.}, et~al.,
  \bibinfo{year}{2003}.
\newblock \bibinfo{title}{Reduction of cholera in bangladeshi villages by
  simple filtration}.
\newblock \bibinfo{journal}{Proceedings of the National Academy of Sciences}
  \bibinfo{volume}{100}, \bibinfo{pages}{1051--1055}.
%Type = Incollection
\bibitem[{Colwell and Spira(1992)}]{colwell1992ecology}
\bibinfo{author}{Colwell, R.R.}, \bibinfo{author}{Spira, W.M.},
  \bibinfo{year}{1992}.
\newblock \bibinfo{title}{The ecology of vibrio cholerae}, in:
  \bibinfo{booktitle}{Cholera}. \bibinfo{publisher}{Springer}, pp.
  \bibinfo{pages}{107--127}.
%Type = Article
\bibitem[{Conner et~al.(2016)Conner, Teschler, Jones and
  Yildiz}]{conner2016staying}
\bibinfo{author}{Conner, J.G.}, \bibinfo{author}{Teschler, J.K.},
  \bibinfo{author}{Jones, C.J.}, \bibinfo{author}{Yildiz, F.H.},
  \bibinfo{year}{2016}.
\newblock \bibinfo{title}{Staying alive: Vibrio cholerae's cycle of
  environmental survival, transmission, and dissemination}.
\newblock \bibinfo{journal}{Virulence mechanisms of bacterial pathogens} ,
  \bibinfo{pages}{593--633}.
%Type = Article
\bibitem[{De~Magny et~al.(2008)De~Magny, Murtugudde, Sapiano, Nizam, Brown,
  Busalacchi, Yunus, Nair, Gil, Lanata et~al.}]{de2008environmental}
\bibinfo{author}{De~Magny, G.C.}, \bibinfo{author}{Murtugudde, R.},
  \bibinfo{author}{Sapiano, M.R.}, \bibinfo{author}{Nizam, A.},
  \bibinfo{author}{Brown, C.W.}, \bibinfo{author}{Busalacchi, A.J.},
  \bibinfo{author}{Yunus, M.}, \bibinfo{author}{Nair, G.B.},
  \bibinfo{author}{Gil, A.I.}, \bibinfo{author}{Lanata, C.F.}, et~al.,
  \bibinfo{year}{2008}.
\newblock \bibinfo{title}{Environmental signatures associated with cholera
  epidemics}.
\newblock \bibinfo{journal}{Proceedings of the National Academy of Sciences}
  \bibinfo{volume}{105}, \bibinfo{pages}{17676--17681}.
%Type = Article
\bibitem[{Di~Capua and Mazzocchi(2017)}]{di2017non}
\bibinfo{author}{Di~Capua, I.}, \bibinfo{author}{Mazzocchi, M.G.},
  \bibinfo{year}{2017}.
\newblock \bibinfo{title}{Non-predatory mortality in mediterranean coastal
  copepods}.
\newblock \bibinfo{journal}{Marine Biology} \bibinfo{volume}{164},
  \bibinfo{pages}{1--12}.
%Type = Article
\bibitem[{Diekmann et~al.(2010)Diekmann, Heesterbeek and
  Roberts}]{diekmann2010construction}
\bibinfo{author}{Diekmann, O.}, \bibinfo{author}{Heesterbeek, J.A.P.},
  \bibinfo{author}{Roberts, M.G.}, \bibinfo{year}{2010}.
\newblock \bibinfo{title}{The construction of next-generation matrices for
  compartmental epidemic models}.
\newblock \bibinfo{journal}{Journal of the royal society interface}
  \bibinfo{volume}{7}, \bibinfo{pages}{873--885}.
%Type = Article
\bibitem[{Dimitrov et~al.(2014)Dimitrov, Troeger, Halloran, Longini and
  Chao}]{dimitrov2014comparative}
\bibinfo{author}{Dimitrov, D.T.}, \bibinfo{author}{Troeger, C.},
  \bibinfo{author}{Halloran, M.E.}, \bibinfo{author}{Longini, I.M.},
  \bibinfo{author}{Chao, D.L.}, \bibinfo{year}{2014}.
\newblock \bibinfo{title}{Comparative effectiveness of different strategies of
  oral cholera vaccination in bangladesh: a modeling study}.
\newblock \bibinfo{journal}{PLoS neglected tropical diseases}
  \bibinfo{volume}{8}, \bibinfo{pages}{e3343}.
%Type = Article
\bibitem[{Dumont et~al.(1975)Dumont, Van~de Velde and Dumont}]{dumont1975dry}
\bibinfo{author}{Dumont, H.J.}, \bibinfo{author}{Van~de Velde, I.},
  \bibinfo{author}{Dumont, S.}, \bibinfo{year}{1975}.
\newblock \bibinfo{title}{The dry weight estimate of biomass in a selection of
  cladocera, copepoda and rotifera from the plankton, periphyton and benthos of
  continental waters}.
\newblock \bibinfo{journal}{Oecologia} \bibinfo{volume}{19},
  \bibinfo{pages}{75--97}.
%Type = Article
\bibitem[{D’Silva et~al.(2012)D’Silva, Anil, Naik and
  D’Costa}]{d2012algal}
\bibinfo{author}{D’Silva, M.S.}, \bibinfo{author}{Anil, A.C.},
  \bibinfo{author}{Naik, R.K.}, \bibinfo{author}{D’Costa, P.M.},
  \bibinfo{year}{2012}.
\newblock \bibinfo{title}{Algal blooms: a perspective from the coasts of
  india}.
\newblock \bibinfo{journal}{Natural hazards} \bibinfo{volume}{63},
  \bibinfo{pages}{1225--1253}.
%Type = Misc
\bibitem[{ECDC(2024)}]{cholera_ECDC}
\bibinfo{author}{ECDC}, \bibinfo{year}{2024}.
\newblock \bibinfo{title}{{European Centre for Disease Prevention and Control:
  Cholera worldwide overview}}.
\newblock
  \bibinfo{howpublished}{\url{https://www.ecdc.europa.eu/en/all-topics-z/cholera/surveillance-and-disease-data/cholera-monthly}}.
%Type = Article
\bibitem[{Enserink(2025)}]{enserink2025cradle}
\bibinfo{author}{Enserink, M.}, \bibinfo{year}{2025}.
\newblock \bibinfo{title}{In the cradle of cholera}.
\newblock \bibinfo{journal}{Science (New York, NY)} \bibinfo{volume}{387},
  \bibinfo{pages}{572--577}.
%Type = Article
\bibitem[{Fewtrell et~al.(2005)Fewtrell, Kaufmann, Kay, Enanoria, Haller and
  Colford}]{fewtrell2005water}
\bibinfo{author}{Fewtrell, L.}, \bibinfo{author}{Kaufmann, R.B.},
  \bibinfo{author}{Kay, D.}, \bibinfo{author}{Enanoria, W.},
  \bibinfo{author}{Haller, L.}, \bibinfo{author}{Colford, J.M.},
  \bibinfo{year}{2005}.
\newblock \bibinfo{title}{Water, sanitation, and hygiene interventions to
  reduce diarrhoea in less developed countries: a systematic review and
  meta-analysis}.
\newblock \bibinfo{journal}{The Lancet infectious diseases}
  \bibinfo{volume}{5}, \bibinfo{pages}{42--52}.
%Type = Article
\bibitem[{Freund et~al.(2006)Freund, Mieruch, Scholze, Wiltshire and
  Feudel}]{freund2006bloom}
\bibinfo{author}{Freund, J.A.}, \bibinfo{author}{Mieruch, S.},
  \bibinfo{author}{Scholze, B.}, \bibinfo{author}{Wiltshire, K.},
  \bibinfo{author}{Feudel, U.}, \bibinfo{year}{2006}.
\newblock \bibinfo{title}{Bloom dynamics in a seasonally forced
  phytoplankton--zooplankton model: trigger mechanisms and timing effects}.
\newblock \bibinfo{journal}{Ecological complexity} \bibinfo{volume}{3},
  \bibinfo{pages}{129--139}.
%Type = Article
\bibitem[{Fung(2014)}]{fung2014cholera}
\bibinfo{author}{Fung, I.C.H.}, \bibinfo{year}{2014}.
\newblock \bibinfo{title}{Cholera transmission dynamic models for public health
  practitioners}.
\newblock \bibinfo{journal}{Emerging themes in epidemiology}
  \bibinfo{volume}{11}, \bibinfo{pages}{1--11}.
%Type = Article
\bibitem[{Handel et~al.(2007)Handel, Longini~Jr and Antia}]{handel2007best}
\bibinfo{author}{Handel, A.}, \bibinfo{author}{Longini~Jr, I.M.},
  \bibinfo{author}{Antia, R.}, \bibinfo{year}{2007}.
\newblock \bibinfo{title}{What is the best control strategy for multiple
  infectious disease outbreaks?}
\newblock \bibinfo{journal}{Proceedings of the Royal Society B: Biological
  Sciences} \bibinfo{volume}{274}, \bibinfo{pages}{833--837}.
%Type = Article
\bibitem[{Hartley et~al.(2005)Hartley, Morris~Jr and
  Smith}]{hartley2005hyperinfectivity}
\bibinfo{author}{Hartley, D.}, \bibinfo{author}{Morris~Jr, J.},
  \bibinfo{author}{Smith, D.}, \bibinfo{year}{2005}.
\newblock \bibinfo{title}{{H}yperinfectivity: {A} {C}ritical {E}lement in the
  {A}bility of v. cholerae to {C}ause {E}pidemics?}
\newblock \bibinfo{journal}{PLoS Med} \bibinfo{volume}{3}, \bibinfo{pages}{e7}.
%Type = Article
\bibitem[{Heidelberg et~al.(2002)Heidelberg, Heidelberg and
  Colwell}]{heidelberg2002bacteria}
\bibinfo{author}{Heidelberg, J.}, \bibinfo{author}{Heidelberg, K.},
  \bibinfo{author}{Colwell, R.}, \bibinfo{year}{2002}.
\newblock \bibinfo{title}{Bacteria of the $\gamma$-subclass proteobacteria
  associated with zooplankton in chesapeake bay}.
\newblock \bibinfo{journal}{Appl. Environ. Microbiol.} \bibinfo{volume}{68},
  \bibinfo{pages}{5498--5507}.
%Type = Article
\bibitem[{Hirst and Ki{\o}rboe(2002)}]{hirst2002mortality}
\bibinfo{author}{Hirst, A.}, \bibinfo{author}{Ki{\o}rboe, T.},
  \bibinfo{year}{2002}.
\newblock \bibinfo{title}{Mortality of marine planktonic copepods: global rates
  and patterns}.
\newblock \bibinfo{journal}{Marine Ecology Progress Series}
  \bibinfo{volume}{230}, \bibinfo{pages}{195--209}.
%Type = Article
\bibitem[{Huq et~al.(2005)Huq, Sack, Nizam, Longini, Nair, Ali, Morris, Khan,
  Siddique, Yunus et~al.}]{huq2005critical}
\bibinfo{author}{Huq, A.}, \bibinfo{author}{Sack, R.B.},
  \bibinfo{author}{Nizam, A.}, \bibinfo{author}{Longini, I.M.},
  \bibinfo{author}{Nair, G.B.}, \bibinfo{author}{Ali, A.},
  \bibinfo{author}{Morris, J.G.}, \bibinfo{author}{Khan, M.H.},
  \bibinfo{author}{Siddique, A.K.}, \bibinfo{author}{Yunus, M.}, et~al.,
  \bibinfo{year}{2005}.
\newblock \bibinfo{title}{Critical factors influencing the occurrence of vibrio
  cholerae in the environment of bangladesh}.
\newblock \bibinfo{journal}{Appl. Environ. Microbiol.} \bibinfo{volume}{71},
  \bibinfo{pages}{4645--4654}.
%Type = Article
\bibitem[{Huq et~al.(1983)Huq, Small, West, Huq, Rahman and
  Colwell}]{huq1983ecological}
\bibinfo{author}{Huq, A.}, \bibinfo{author}{Small, E.B.},
  \bibinfo{author}{West, P.A.}, \bibinfo{author}{Huq, M.I.},
  \bibinfo{author}{Rahman, R.}, \bibinfo{author}{Colwell, R.R.},
  \bibinfo{year}{1983}.
\newblock \bibinfo{title}{Ecological relationships between vibrio cholerae and
  planktonic crustacean copepods.}
\newblock \bibinfo{journal}{Appl. Environ. Microbiol.} \bibinfo{volume}{45},
  \bibinfo{pages}{275--283}.
%Type = Article
\bibitem[{Huq et~al.(1984)Huq, West, Small, Huq and Colwell}]{huq1984influence}
\bibinfo{author}{Huq, A.}, \bibinfo{author}{West, P.A.},
  \bibinfo{author}{Small, E.B.}, \bibinfo{author}{Huq, M.I.},
  \bibinfo{author}{Colwell, R.R.}, \bibinfo{year}{1984}.
\newblock \bibinfo{title}{Influence of water temperature, salinity, and ph on
  survival and growth of toxigenic vibrio cholerae serovar 01 associated with
  live copepods in laboratory microcosms.}
\newblock \bibinfo{journal}{Appl. Environ. Microbiol.} \bibinfo{volume}{48},
  \bibinfo{pages}{420--424}.
%Type = Article
\bibitem[{Huq et~al.(1996)Huq, Xu, Chowdhury, Islam, Montilla and
  Colwell}]{huq1996simple}
\bibinfo{author}{Huq, A.}, \bibinfo{author}{Xu, B.},
  \bibinfo{author}{Chowdhury, M.}, \bibinfo{author}{Islam, M.S.},
  \bibinfo{author}{Montilla, R.}, \bibinfo{author}{Colwell, R.R.},
  \bibinfo{year}{1996}.
\newblock \bibinfo{title}{A simple filtration method to remove
  plankton-associated vibrio cholerae in raw water supplies in developing
  countries}.
\newblock \bibinfo{journal}{Applied and environmental microbiology}
  \bibinfo{volume}{62}, \bibinfo{pages}{2508--2512}.
%Type = Article
\bibitem[{Huq et~al.(2010)Huq, Yunus, Sohel, Bhuiya, Emch, Luby, Russek-Cohen,
  Nair, Sack and Colwell}]{huq2010simple}
\bibinfo{author}{Huq, A.}, \bibinfo{author}{Yunus, M.}, \bibinfo{author}{Sohel,
  S.S.}, \bibinfo{author}{Bhuiya, A.}, \bibinfo{author}{Emch, M.},
  \bibinfo{author}{Luby, S.P.}, \bibinfo{author}{Russek-Cohen, E.},
  \bibinfo{author}{Nair, G.B.}, \bibinfo{author}{Sack, R.B.},
  \bibinfo{author}{Colwell, R.R.}, \bibinfo{year}{2010}.
\newblock \bibinfo{title}{Simple sari cloth filtration of water is sustainable
  and continues to protect villagers from cholera in matlab, bangladesh}.
\newblock \bibinfo{journal}{MBio} \bibinfo{volume}{1},
  \bibinfo{pages}{e00034--10}.
%Type = Article
\bibitem[{Ilic and Ilic(2023)}]{ilic2023global}
\bibinfo{author}{Ilic, I.}, \bibinfo{author}{Ilic, M.}, \bibinfo{year}{2023}.
\newblock \bibinfo{title}{Global patterns of trends in cholera mortality}.
\newblock \bibinfo{journal}{Tropical Medicine and Infectious Disease}
  \bibinfo{volume}{8}, \bibinfo{pages}{169}.
%Type = Article
\bibitem[{Islam et~al.(1994)Islam, Drasar and Sack}]{islam1994probable}
\bibinfo{author}{Islam, M.S.}, \bibinfo{author}{Drasar, B.S.},
  \bibinfo{author}{Sack, R.B.}, \bibinfo{year}{1994}.
\newblock \bibinfo{title}{Probable role of blue-green algae in maintaining
  endemicity and seasonality of cholera in bangladesh: a hypothesis}.
\newblock \bibinfo{journal}{Journal of diarrhoeal diseases research} ,
  \bibinfo{pages}{245--256}.
%Type = Article
\bibitem[{Islam et~al.(2015)Islam, Islam, Mahmud, Cairncross, Clemens and
  Collins}]{islam2015role}
\bibinfo{author}{Islam, M.S.}, \bibinfo{author}{Islam, M.S.},
  \bibinfo{author}{Mahmud, Z.H.}, \bibinfo{author}{Cairncross, S.},
  \bibinfo{author}{Clemens, J.D.}, \bibinfo{author}{Collins, A.E.},
  \bibinfo{year}{2015}.
\newblock \bibinfo{title}{Role of phytoplankton in maintaining endemicity and
  seasonality of cholera in bangladesh}.
\newblock \bibinfo{journal}{Transactions of The Royal Society of Tropical
  Medicine and Hygiene} \bibinfo{volume}{109}, \bibinfo{pages}{572--578}.
%Type = Article
\bibitem[{Jensen et~al.(2006)Jensen, Faruque, Mekalanos and
  Levin}]{jensen2006modeling}
\bibinfo{author}{Jensen, M.A.}, \bibinfo{author}{Faruque, S.M.},
  \bibinfo{author}{Mekalanos, J.J.}, \bibinfo{author}{Levin, B.R.},
  \bibinfo{year}{2006}.
\newblock \bibinfo{title}{Modeling the role of bacteriophage in the control of
  cholera outbreaks}.
\newblock \bibinfo{journal}{Proceedings of the national academy of Sciences}
  \bibinfo{volume}{103}, \bibinfo{pages}{4652--4657}.
%Type = Article
\bibitem[{Jutla et~al.(2012)Jutla, Akanda and Islam}]{jutla2012satellite}
\bibinfo{author}{Jutla, A.S.}, \bibinfo{author}{Akanda, A.S.},
  \bibinfo{author}{Islam, S.}, \bibinfo{year}{2012}.
\newblock \bibinfo{title}{Satellite remote sensing of space--time plankton
  variability in the bay of bengal: Connections to cholera outbreaks}.
\newblock \bibinfo{journal}{Remote sensing of environment}
  \bibinfo{volume}{123}, \bibinfo{pages}{196--206}.
%Type = Book
\bibitem[{Keeling and Rohani(2008)}]{keeling2008modeling}
\bibinfo{author}{Keeling, M.}, \bibinfo{author}{Rohani, P.},
  \bibinfo{year}{2008}.
\newblock \bibinfo{title}{{M}odeling {I}nfectious {D}iseases in {H}umans and
  {A}nimals}.
\newblock \bibinfo{publisher}{Princeton University Press}.
%Type = Article
\bibitem[{King et~al.(2008)King, Ionides, Pascual and
  Bouma}]{king2008inapparent}
\bibinfo{author}{King, A.}, \bibinfo{author}{Ionides, E.},
  \bibinfo{author}{Pascual, M.}, \bibinfo{author}{Bouma, M.},
  \bibinfo{year}{2008}.
\newblock \bibinfo{title}{Inapparent infections and cholera dynamics}.
\newblock \bibinfo{journal}{Nature} \bibinfo{volume}{454},
  \bibinfo{pages}{877--880}.
%Type = Article
\bibitem[{Kirn et~al.(2005)Kirn, Jude and Taylor}]{kirn2005colonization}
\bibinfo{author}{Kirn, T.J.}, \bibinfo{author}{Jude, B.A.},
  \bibinfo{author}{Taylor, R.K.}, \bibinfo{year}{2005}.
\newblock \bibinfo{title}{A colonization factor links vibrio cholerae
  environmental survival and human infection}.
\newblock \bibinfo{journal}{Nature} \bibinfo{volume}{438},
  \bibinfo{pages}{863--866}.
%Type = Article
\bibitem[{Koelle et~al.(2005)Koelle, Rod{\'o}, Pascual, Yunus and
  Mostafa}]{koelle2005refractory}
\bibinfo{author}{Koelle, K.}, \bibinfo{author}{Rod{\'o}, X.},
  \bibinfo{author}{Pascual, M.}, \bibinfo{author}{Yunus, M.},
  \bibinfo{author}{Mostafa, G.}, \bibinfo{year}{2005}.
\newblock \bibinfo{title}{Refractory periods and climate forcing in cholera
  dynamics}.
\newblock \bibinfo{journal}{Nature} \bibinfo{volume}{436},
  \bibinfo{pages}{696--700}.
%Type = Article
\bibitem[{Kolaye et~al.(2019)Kolaye, Bowong, Houe, Aziz-Alaoui and
  Cadivel}]{kolaye2019mathematical}
\bibinfo{author}{Kolaye, G.}, \bibinfo{author}{Bowong, S.},
  \bibinfo{author}{Houe, R.}, \bibinfo{author}{Aziz-Alaoui, M.A.},
  \bibinfo{author}{Cadivel, M.}, \bibinfo{year}{2019}.
\newblock \bibinfo{title}{Mathematical assessment of the role of environmental
  factors on the dynamical transmission of cholera}.
\newblock \bibinfo{journal}{Communications in Nonlinear Science and Numerical
  Simulation} \bibinfo{volume}{67}, \bibinfo{pages}{203--222}.
%Type = Article
\bibitem[{Lipp et~al.(2002)Lipp, Huq and Colwell}]{lipp2002effects}
\bibinfo{author}{Lipp, E.K.}, \bibinfo{author}{Huq, A.},
  \bibinfo{author}{Colwell, R.R.}, \bibinfo{year}{2002}.
\newblock \bibinfo{title}{Effects of global climate on infectious disease: the
  cholera model}.
\newblock \bibinfo{journal}{Clinical microbiology reviews}
  \bibinfo{volume}{15}, \bibinfo{pages}{757--770}.
%Type = Article
\bibitem[{Liz{\'a}rraga-Partida et~al.(2009)Liz{\'a}rraga-Partida,
  Mendez-G{\'o}mez, Rivas-Monta{\~n}o, Vargas-Hern{\'a}ndez,
  Portillo-L{\'o}pez, Gonz{\'a}lez-Ram{\'\i}rez, Huq and
  Colwell}]{lizarraga2009association}
\bibinfo{author}{Liz{\'a}rraga-Partida, M.}, \bibinfo{author}{Mendez-G{\'o}mez,
  E.}, \bibinfo{author}{Rivas-Monta{\~n}o, A.},
  \bibinfo{author}{Vargas-Hern{\'a}ndez, E.},
  \bibinfo{author}{Portillo-L{\'o}pez, A.},
  \bibinfo{author}{Gonz{\'a}lez-Ram{\'\i}rez, A.}, \bibinfo{author}{Huq, A.},
  \bibinfo{author}{Colwell, R.}, \bibinfo{year}{2009}.
\newblock \bibinfo{title}{Association of vibrio cholerae with plankton in
  coastal areas of mexico}.
\newblock \bibinfo{journal}{Environmental microbiology} \bibinfo{volume}{11},
  \bibinfo{pages}{201--208}.
%Type = Article
\bibitem[{Lutz et~al.(2013)Lutz, Erken, Noorian, Sun and
  McDougald}]{lutz2013environmental}
\bibinfo{author}{Lutz, C.}, \bibinfo{author}{Erken, M.},
  \bibinfo{author}{Noorian, P.}, \bibinfo{author}{Sun, S.},
  \bibinfo{author}{McDougald, D.}, \bibinfo{year}{2013}.
\newblock \bibinfo{title}{Environmental reservoirs and mechanisms of
  persistence of vibrio cholerae}.
\newblock \bibinfo{journal}{Frontiers in microbiology} \bibinfo{volume}{4},
  \bibinfo{pages}{375}.
%Type = Article
\bibitem[{Constantin~de Magny et~al.(2014)Constantin~de Magny, Hasan and
  Roche}]{constantin2014community}
\bibinfo{author}{Constantin~de Magny, G.}, \bibinfo{author}{Hasan, N.A.},
  \bibinfo{author}{Roche, B.}, \bibinfo{year}{2014}.
\newblock \bibinfo{title}{How community ecology can improve our understanding
  of cholera dynamics}.
\newblock \bibinfo{journal}{Frontiers in microbiology} \bibinfo{volume}{5},
  \bibinfo{pages}{137}.
%Type = Article
\bibitem[{Constantin~de Magny et~al.(2011)Constantin~de Magny, Mozumder, Grim,
  Hasan, Naser, Alam, Sack, Huq and Colwell}]{de2011role}
\bibinfo{author}{Constantin~de Magny, G.}, \bibinfo{author}{Mozumder, P.K.},
  \bibinfo{author}{Grim, C.J.}, \bibinfo{author}{Hasan, N.A.},
  \bibinfo{author}{Naser, M.N.}, \bibinfo{author}{Alam, M.},
  \bibinfo{author}{Sack, R.B.}, \bibinfo{author}{Huq, A.},
  \bibinfo{author}{Colwell, R.R.}, \bibinfo{year}{2011}.
\newblock \bibinfo{title}{Role of zooplankton diversity in vibrio cholerae
  population dynamics and in the incidence of cholera in the bangladesh
  sundarbans}.
\newblock \bibinfo{journal}{Appl. Environ. Microbiol.} \bibinfo{volume}{77},
  \bibinfo{pages}{6125--6132}.
%Type = Article
\bibitem[{Maity et~al.(2023)Maity, Saha, Ghosh and
  Chattopadhyay}]{maity2023model}
\bibinfo{author}{Maity, B.}, \bibinfo{author}{Saha, B.},
  \bibinfo{author}{Ghosh, I.}, \bibinfo{author}{Chattopadhyay, J.},
  \bibinfo{year}{2023}.
\newblock \bibinfo{title}{Model-based estimation of expected time to cholera
  extinction in lusaka, zambia}.
\newblock \bibinfo{journal}{Bulletin of Mathematical Biology}
  \bibinfo{volume}{85}, \bibinfo{pages}{55}.
%Type = Article
\bibitem[{Marino et~al.(2008)Marino, Hogue, Ray and
  Kirschner}]{marino2008methodology}
\bibinfo{author}{Marino, S.}, \bibinfo{author}{Hogue, I.B.},
  \bibinfo{author}{Ray, C.J.}, \bibinfo{author}{Kirschner, D.E.},
  \bibinfo{year}{2008}.
\newblock \bibinfo{title}{{A} methodology for performing global uncertainty and
  sensitivity analysis in systems biology}.
\newblock \bibinfo{journal}{J. Theor. Biol.} \bibinfo{volume}{254},
  \bibinfo{pages}{178--196}.
%Type = Book
\bibitem[{Martcheva(2015)}]{martcheva2015introduction}
\bibinfo{author}{Martcheva, M.}, \bibinfo{year}{2015}.
\newblock \bibinfo{title}{An introduction to mathematical epidemiology}.
  volume~\bibinfo{volume}{61}.
\newblock \bibinfo{publisher}{Springer}.
%Type = Article
\bibitem[{Martinelli~Filho et~al.(2010)Martinelli~Filho, Lopes, Rivera and
  Colwell}]{martinelli2010vibrio}
\bibinfo{author}{Martinelli~Filho, J.E.}, \bibinfo{author}{Lopes, R.M.},
  \bibinfo{author}{Rivera, I.N.}, \bibinfo{author}{Colwell, R.R.},
  \bibinfo{year}{2010}.
\newblock \bibinfo{title}{Vibrio cholerae o1 detection in estuarine and coastal
  zooplankton}.
\newblock \bibinfo{journal}{Journal of plankton research} \bibinfo{volume}{33},
  \bibinfo{pages}{51--62}.
%Type = Article
\bibitem[{Miller~Neilan et~al.(2010)Miller~Neilan, Schaefer, Gaff, Fister and
  Lenhart}]{miller2010modeling}
\bibinfo{author}{Miller~Neilan, R.L.}, \bibinfo{author}{Schaefer, E.},
  \bibinfo{author}{Gaff, H.}, \bibinfo{author}{Fister, K.R.},
  \bibinfo{author}{Lenhart, S.}, \bibinfo{year}{2010}.
\newblock \bibinfo{title}{Modeling optimal intervention strategies for
  cholera}.
\newblock \bibinfo{journal}{Bulletin of mathematical biology}
  \bibinfo{volume}{72}, \bibinfo{pages}{2004--2018}.
%Type = Article
\bibitem[{Mukandavire et~al.(2011)Mukandavire, Liao, Wang, Gaff, Smith and
  Morris}]{mukandavire2011estimating}
\bibinfo{author}{Mukandavire, Z.}, \bibinfo{author}{Liao, S.},
  \bibinfo{author}{Wang, J.}, \bibinfo{author}{Gaff, H.},
  \bibinfo{author}{Smith, D.}, \bibinfo{author}{Morris, J.G.},
  \bibinfo{year}{2011}.
\newblock \bibinfo{title}{{E}stimating the reproductive numbers for the
  2008--2009 cholera outbreaks in {Z}imbabwe}.
\newblock \bibinfo{journal}{Proc. Natl. Acad. Sci. U.S.A.}
  \bibinfo{volume}{108}, \bibinfo{pages}{8767--8772}.
%Type = Article
\bibitem[{Neilan et~al.(2010)Neilan, Schaefer, Gaff, Fister and
  Lenhart}]{neilan2010modeling}
\bibinfo{author}{Neilan, R.}, \bibinfo{author}{Schaefer, E.},
  \bibinfo{author}{Gaff, H.}, \bibinfo{author}{Fister, K.},
  \bibinfo{author}{Lenhart, S.}, \bibinfo{year}{2010}.
\newblock \bibinfo{title}{Modeling optimal intervention strategies for
  cholera}.
\newblock \bibinfo{journal}{Bull. Math. Biol.} \bibinfo{volume}{72},
  \bibinfo{pages}{2004--2018}.
%Type = Article
\bibitem[{Nguyen et~al.(2023)Nguyen, Freedman, Ozbay and
  Levin}]{nguyen2023fundamental}
\bibinfo{author}{Nguyen, M.M.}, \bibinfo{author}{Freedman, A.S.},
  \bibinfo{author}{Ozbay, S.A.}, \bibinfo{author}{Levin, S.A.},
  \bibinfo{year}{2023}.
\newblock \bibinfo{title}{Fundamental bound on epidemic overshoot in the sir
  model}.
\newblock \bibinfo{journal}{Journal of the Royal Society Interface}
  \bibinfo{volume}{20}, \bibinfo{pages}{20230322}.
%Type = Article
\bibitem[{Perera et~al.(2022)Perera, Fujiyoshi, Nishiuchi, Nakai and
  Maruyama}]{perera2022zooplankton}
\bibinfo{author}{Perera, I.U.}, \bibinfo{author}{Fujiyoshi, S.},
  \bibinfo{author}{Nishiuchi, Y.}, \bibinfo{author}{Nakai, T.},
  \bibinfo{author}{Maruyama, F.}, \bibinfo{year}{2022}.
\newblock \bibinfo{title}{Zooplankton act as cruise ships promoting the
  survival and pathogenicity of pathogenic bacteria}.
\newblock \bibinfo{journal}{Microbiology and Immunology} \bibinfo{volume}{66},
  \bibinfo{pages}{564--578}.
%Type = Article
\bibitem[{Rawlings et~al.(2007)Rawlings, Ruiz and
  Colwell}]{rawlings2007association}
\bibinfo{author}{Rawlings, T.K.}, \bibinfo{author}{Ruiz, G.M.},
  \bibinfo{author}{Colwell, R.R.}, \bibinfo{year}{2007}.
\newblock \bibinfo{title}{Association of vibrio cholerae o1 el tor and o139
  bengal with the copepods acartia tonsa and eurytemora affinis}.
\newblock \bibinfo{journal}{Appl. Environ. Microbiol.} \bibinfo{volume}{73},
  \bibinfo{pages}{7926--7933}.
%Type = Article
\bibitem[{Righetto et~al.(2012)Righetto, Casagrandi, Bertuzzo, Mari, Gatto,
  Rodriguez-Iturbe and Rinaldo}]{righetto2012role}
\bibinfo{author}{Righetto, L.}, \bibinfo{author}{Casagrandi, R.},
  \bibinfo{author}{Bertuzzo, E.}, \bibinfo{author}{Mari, L.},
  \bibinfo{author}{Gatto, M.}, \bibinfo{author}{Rodriguez-Iturbe, I.},
  \bibinfo{author}{Rinaldo, A.}, \bibinfo{year}{2012}.
\newblock \bibinfo{title}{The role of aquatic reservoir fluctuations in
  long-term cholera patterns}.
\newblock \bibinfo{journal}{Epidemics} \bibinfo{volume}{4},
  \bibinfo{pages}{33--42}.
%Type = Article
\bibitem[{Rosenzweig(1971)}]{rosenzweig1971paradox}
\bibinfo{author}{Rosenzweig, M.L.}, \bibinfo{year}{1971}.
\newblock \bibinfo{title}{Paradox of enrichment: destabilization of
  exploitation ecosystems in ecological time}.
\newblock \bibinfo{journal}{Science} \bibinfo{volume}{171},
  \bibinfo{pages}{385--387}.
%Type = Article
\bibitem[{Sack et~al.(2021)Sack, Debes, Ateudjieu, Bwire, Ali, Ngwa, Mwaba,
  Chilengi, Orach, Boru et~al.}]{sack2021contrasting}
\bibinfo{author}{Sack, D.A.}, \bibinfo{author}{Debes, A.K.},
  \bibinfo{author}{Ateudjieu, J.}, \bibinfo{author}{Bwire, G.},
  \bibinfo{author}{Ali, M.}, \bibinfo{author}{Ngwa, M.C.},
  \bibinfo{author}{Mwaba, J.}, \bibinfo{author}{Chilengi, R.},
  \bibinfo{author}{Orach, C.C.}, \bibinfo{author}{Boru, W.}, et~al.,
  \bibinfo{year}{2021}.
\newblock \bibinfo{title}{Contrasting epidemiology of cholera in bangladesh and
  africa}.
\newblock \bibinfo{journal}{The Journal of infectious diseases}
  \bibinfo{volume}{224}, \bibinfo{pages}{S701--S709}.
%Type = Article
\bibitem[{Saltelli et~al.(2004)Saltelli, Tarantola, Campolongo, Ratto
  et~al.}]{saltelli2004sensitivity}
\bibinfo{author}{Saltelli, A.}, \bibinfo{author}{Tarantola, S.},
  \bibinfo{author}{Campolongo, F.}, \bibinfo{author}{Ratto, M.}, et~al.,
  \bibinfo{year}{2004}.
\newblock \bibinfo{title}{Sensitivity analysis in practice: a guide to
  assessing scientific models}.
\newblock \bibinfo{journal}{Chichester, England} .
%Type = Article
\bibitem[{Sanches et~al.(2011)Sanches, Ferreira and Kraenkel}]{sanches2011role}
\bibinfo{author}{Sanches, R.P.}, \bibinfo{author}{Ferreira, C.P.},
  \bibinfo{author}{Kraenkel, R.A.}, \bibinfo{year}{2011}.
\newblock \bibinfo{title}{The role of immunity and seasonality in cholera
  epidemics}.
\newblock \bibinfo{journal}{Bulletin of mathematical biology}
  \bibinfo{volume}{73}, \bibinfo{pages}{2916--2931}.
%Type = Article
\bibitem[{Scheffer(1991)}]{scheffer1991fish}
\bibinfo{author}{Scheffer, M.}, \bibinfo{year}{1991}.
\newblock \bibinfo{title}{Fish and nutrients interplay determines algal
  biomass: a minimal model}.
\newblock \bibinfo{journal}{Oikos} , \bibinfo{pages}{271--282}.
%Type = Article
\bibitem[{Seeligmann et~al.(2008)Seeligmann, Mirande, Tracanna, Silva, Aulet,
  Cecilia and Binsztein}]{phyto2008seeligmann}
\bibinfo{author}{Seeligmann, C.T.}, \bibinfo{author}{Mirande, V.},
  \bibinfo{author}{Tracanna, B.C.}, \bibinfo{author}{Silva, C.},
  \bibinfo{author}{Aulet, O.}, \bibinfo{author}{Cecilia, M.},
  \bibinfo{author}{Binsztein, N.}, \bibinfo{year}{2008}.
\newblock \bibinfo{title}{{Phytoplankton-linked viable non-culturable Vibrio
  cholerae O1 (VNC) from rivers in Tucumán, Argentina}}.
\newblock \bibinfo{journal}{Journal of Plankton Research} \bibinfo{volume}{30},
  \bibinfo{pages}{367--377}.
%Type = Article
\bibitem[{Shackleton et~al.(2024)Shackleton, Memon, Nichols, Phalkey and
  Chen}]{shackleton2024mechanisms}
\bibinfo{author}{Shackleton, D.}, \bibinfo{author}{Memon, F.A.},
  \bibinfo{author}{Nichols, G.}, \bibinfo{author}{Phalkey, R.},
  \bibinfo{author}{Chen, A.S.}, \bibinfo{year}{2024}.
\newblock \bibinfo{title}{Mechanisms of cholera transmission via environment in
  india and bangladesh: state of the science review}.
\newblock \bibinfo{journal}{Reviews on Environmental Health}
  \bibinfo{volume}{39}, \bibinfo{pages}{313--329}.
%Type = Article
\bibitem[{da~Silva et~al.(2020)da~Silva, de~Castro~Melo, Neumann-Leit{\~a}o and
  de~Melo~J{\'u}nior}]{da2020non}
\bibinfo{author}{da~Silva, A.J.}, \bibinfo{author}{de~Castro~Melo, P.A.M.},
  \bibinfo{author}{Neumann-Leit{\~a}o, S.},
  \bibinfo{author}{de~Melo~J{\'u}nior, M.}, \bibinfo{year}{2020}.
\newblock \bibinfo{title}{Non-predatory mortality of planktonic copepods in a
  reef area influenced by estuarine plume}.
\newblock \bibinfo{journal}{Marine Environmental Research}
  \bibinfo{volume}{159}, \bibinfo{pages}{105024}.
%Type = Article
\bibitem[{Sobsey et~al.(2008)Sobsey, Stauber, Casanova, Brown and
  Elliott}]{sobsey2008point}
\bibinfo{author}{Sobsey, M.D.}, \bibinfo{author}{Stauber, C.E.},
  \bibinfo{author}{Casanova, L.M.}, \bibinfo{author}{Brown, J.M.},
  \bibinfo{author}{Elliott, M.A.}, \bibinfo{year}{2008}.
\newblock \bibinfo{title}{Point of use household drinking water filtration: a
  practical, effective solution for providing sustained access to safe drinking
  water in the developing world}.
\newblock \bibinfo{journal}{Environmental science \& technology}
  \bibinfo{volume}{42}, \bibinfo{pages}{4261--4267}.
%Type = Article
\bibitem[{Sochard et~al.(1979)Sochard, Wilson, Austin and
  Colwell}]{sochard1979bacteria}
\bibinfo{author}{Sochard, M.}, \bibinfo{author}{Wilson, D.},
  \bibinfo{author}{Austin, B.}, \bibinfo{author}{Colwell, R.},
  \bibinfo{year}{1979}.
\newblock \bibinfo{title}{Bacteria associated with the surface and gut of
  marine copepods}.
\newblock \bibinfo{journal}{Applied and Environmental Microbiology}
  \bibinfo{volume}{37}, \bibinfo{pages}{750--759}.
%Type = Article
\bibitem[{Tang et~al.(2010)Tang, Turk and Grossart}]{tang2010linkage}
\bibinfo{author}{Tang, K.W.}, \bibinfo{author}{Turk, V.},
  \bibinfo{author}{Grossart, H.P.}, \bibinfo{year}{2010}.
\newblock \bibinfo{title}{Linkage between crustacean zooplankton and aquatic
  bacteria}.
\newblock \bibinfo{journal}{Aquatic Microbial Ecology} \bibinfo{volume}{61},
  \bibinfo{pages}{261--277}.
%Type = Article
\bibitem[{Taylor et~al.(2015)Taylor, Kahawita, Cairncross and
  Ensink}]{taylor2015impact}
\bibinfo{author}{Taylor, D.L.}, \bibinfo{author}{Kahawita, T.M.},
  \bibinfo{author}{Cairncross, S.}, \bibinfo{author}{Ensink, J.H.},
  \bibinfo{year}{2015}.
\newblock \bibinfo{title}{The impact of water, sanitation and hygiene
  interventions to control cholera: a systematic review}.
\newblock \bibinfo{journal}{PLoS one} \bibinfo{volume}{10},
  \bibinfo{pages}{e0135676}.
%Type = Article
\bibitem[{Tien and Earn(2010)}]{tien2010multiple}
\bibinfo{author}{Tien, J.H.}, \bibinfo{author}{Earn, D.J.},
  \bibinfo{year}{2010}.
\newblock \bibinfo{title}{Multiple transmission pathways and disease dynamics
  in a waterborne pathogen model}.
\newblock \bibinfo{journal}{Bulletin of mathematical biology}
  \bibinfo{volume}{72}, \bibinfo{pages}{1506--1533}.
%Type = Article
\bibitem[{Troeger et~al.(2018)Troeger, Blacker, Khalil, Rao, Cao, Zimsen,
  Albertson, Stanaway, Deshpande, Abebe et~al.}]{troeger2018estimates}
\bibinfo{author}{Troeger, C.}, \bibinfo{author}{Blacker, B.F.},
  \bibinfo{author}{Khalil, I.A.}, \bibinfo{author}{Rao, P.C.},
  \bibinfo{author}{Cao, S.}, \bibinfo{author}{Zimsen, S.R.},
  \bibinfo{author}{Albertson, S.B.}, \bibinfo{author}{Stanaway, J.D.},
  \bibinfo{author}{Deshpande, A.}, \bibinfo{author}{Abebe, Z.}, et~al.,
  \bibinfo{year}{2018}.
\newblock \bibinfo{title}{Estimates of the global, regional, and national
  morbidity, mortality, and aetiologies of diarrhoea in 195 countries: a
  systematic analysis for the global burden of disease study 2016}.
\newblock \bibinfo{journal}{The Lancet Infectious Diseases}
  \bibinfo{volume}{18}, \bibinfo{pages}{1211--1228}.
%Type = Article
\bibitem[{Turner et~al.(2009)Turner, Good, Cole and Lipp}]{turner2009plankton}
\bibinfo{author}{Turner, J.W.}, \bibinfo{author}{Good, B.},
  \bibinfo{author}{Cole, D.}, \bibinfo{author}{Lipp, E.K.},
  \bibinfo{year}{2009}.
\newblock \bibinfo{title}{Plankton composition and environmental factors
  contribute to vibrio seasonality}.
\newblock \bibinfo{journal}{The ISME journal} \bibinfo{volume}{3},
  \bibinfo{pages}{1082}.
%Type = Article
\bibitem[{Venkatesan(2024)}]{venkatesan2024new}
\bibinfo{author}{Venkatesan, P.}, \bibinfo{year}{2024}.
\newblock \bibinfo{title}{New measures to tackle the global cholera surge}.
\newblock \bibinfo{journal}{The Lancet Microbe} .
%Type = Article
\bibitem[{Vezzulli et~al.(2010)Vezzulli, Pruzzo, Huq and
  Colwell}]{vezzulli2010environmental}
\bibinfo{author}{Vezzulli, L.}, \bibinfo{author}{Pruzzo, C.},
  \bibinfo{author}{Huq, A.}, \bibinfo{author}{Colwell, R.R.},
  \bibinfo{year}{2010}.
\newblock \bibinfo{title}{Environmental reservoirs of vibrio cholerae and their
  role in cholera}.
\newblock \bibinfo{journal}{Environmental microbiology reports}
  \bibinfo{volume}{2}, \bibinfo{pages}{27--33}.
%Type = Article
\bibitem[{Walther and Ewald(2004)}]{walther2004pathogen}
\bibinfo{author}{Walther, B.A.}, \bibinfo{author}{Ewald, P.W.},
  \bibinfo{year}{2004}.
\newblock \bibinfo{title}{Pathogen survival in the external environment and the
  evolution of virulence}.
\newblock \bibinfo{journal}{Biological Reviews} \bibinfo{volume}{79},
  \bibinfo{pages}{849--869}.
%Type = Misc
\bibitem[{WHO(2022)}]{cholera_africa}
\bibinfo{author}{WHO}, \bibinfo{year}{2022}.
\newblock \bibinfo{title}{{WHO and Partners revamp war against cholera in
  Africa}}.
\newblock
  \bibinfo{howpublished}{\url{https://www.afro.who.int/countries/united-republic-of-tanzania/news/who-and-partners-revamp-war-against-cholera-africa}}.
%Type = Misc
\bibitem[{WHO(2023)}]{cholera_global}
\bibinfo{author}{WHO}, \bibinfo{year}{2023}.
\newblock \bibinfo{title}{{Cholera – Global situation}}.
\newblock
  \bibinfo{howpublished}{\url{https://www.who.int/emergencies/disease-outbreak-news/item/2023-DON437}}.
%Type = Misc
\bibitem[{WHO(2024)}]{cholera12}
\bibinfo{author}{WHO}, \bibinfo{year}{2024}.
\newblock \bibinfo{title}{{Multi-country outbreak of cholera, External
  situation report 12 - 14 March 2024}}.
\newblock
  \bibinfo{howpublished}{\url{https://www.who.int/publications/m/item/multi-country-outbreak-of-cholera--external-situation-report--12---14-march-2024}}.
%Type = Article
\bibitem[{Xu et~al.(2024)Xu, Zou, Dent, Wiens, Malembaka, Bwire, Okitayemba,
  Hampton, Azman and Lee}]{xu2024enhanced}
\bibinfo{author}{Xu, H.}, \bibinfo{author}{Zou, K.}, \bibinfo{author}{Dent,
  J.}, \bibinfo{author}{Wiens, K.E.}, \bibinfo{author}{Malembaka, E.B.},
  \bibinfo{author}{Bwire, G.}, \bibinfo{author}{Okitayemba, P.W.},
  \bibinfo{author}{Hampton, L.M.}, \bibinfo{author}{Azman, A.S.},
  \bibinfo{author}{Lee, E.C.}, \bibinfo{year}{2024}.
\newblock \bibinfo{title}{Enhanced cholera surveillance to improve vaccination
  campaign efficiency}.
\newblock \bibinfo{journal}{Nature Medicine} , \bibinfo{pages}{1--7}.
%Type = Article
\bibitem[{Yang et~al.(1996)Yang, Chen and Chen}]{yang1996permanence}
\bibinfo{author}{Yang, X.}, \bibinfo{author}{Chen, L.}, \bibinfo{author}{Chen,
  J.}, \bibinfo{year}{1996}.
\newblock \bibinfo{title}{Permanence and positive periodic solution for the
  single-species nonautonomous delay diffusive models}.
\newblock \bibinfo{journal}{Computers \& Mathematics with Applications}
  \bibinfo{volume}{32}, \bibinfo{pages}{109--116}.

\end{thebibliography}
\bibliographystyle{elsarticle-harv}
\end{document}